\documentclass[a4paper]{article}
 \usepackage{fullpage}
  \usepackage{amsthm,amsmath,amsfonts,amssymb}
  \usepackage{subfigure}
  \usepackage{verbatim}
  \usepackage{tikz}
  \usepackage[T1]{fontenc}
  \usepackage[utf8]{inputenc}
  \usepackage{xspace}
  \usepackage{enumitem}
  \usepackage{listings}
  \usepackage{enumitem}
  \setitemize{itemsep=1pt, topsep=2pt, parsep=0pt, partopsep=0pt}
  \setenumerate{itemsep=1pt, topsep=2pt, parsep=0pt, partopsep=0pt}
  \usepackage{array}
  \usepackage{authblk}
  \usepackage{todonotes}
  \usepackage[noend, noline, ruled]{algorithm2e}
\usepackage{xcolor,stackengine,graphicx,scalerel}

  \usepackage[hidelinks]{hyperref}
  \usepackage[capitalise]{cleveref}

  \newtheorem{theorem}{Theorem}[section]
  \newtheorem{lemma}[theorem]{Lemma}
  \newtheorem{fact}[theorem]{Fact}
  \newtheorem{corollary}[theorem]{Corollary}
  \newtheorem{observation}[theorem]{Observation}
  
  \theoremstyle{remark}
  \newtheorem{example}[theorem]{Example}
 
  \usetikzlibrary{decorations.pathreplacing,calc,patterns}

  \crefname{fact}{Fact}{Facts}
  \crefname{observation}{Observation}{Observations}

  \newcommand{\dziura}{\diamondsuit}
  
  \newcommand{\Fr}{\mathcal{F}}
    \newcommand{\pred}{\operatorname{Left}}
  \renewcommand{\succ}{\operatorname{Right}}
  \newcommand{\al}{\mathsf{a}}
  \newcommand{\bl}{\mathsf{b}}
  \newcommand{\St}{\mathrm{St}}
\newcommand{\p}{\mathsf{p}}
\newcommand{\q}{\mathsf{q}}
\newcommand{\fr}{\mathrm{fr}}

  \newcommand{\floor}[1]{\left\lfloor #1 \right\rfloor}
  \newcommand{\ceil}[1]{\left\lceil #1 \right\rceil}
  \newcommand{\Oh}{\mathcal{O}}

  \title{On Periodicity Lemma for Partial Words\thanks{Supported by the Polish National Science
    Center, grant no 2014/13/B/ST6/00770.}}
  \author{Tomasz Kociumaka}
  \author{Jakub Radoszewski}
  \author{Wojciech Rytter}
  \author{Tomasz Waleń}
  \affil{\normalsize Faculty~of Mathematics, Informatics and Mechanics,\\
    University of Warsaw, Warsaw, Poland\\
    \texttt{[kociumaka,jrad,rytter,walen]@mimuw.edu.pl}}

\date{\vspace{-1cm}}

\begin{document}
  \maketitle
\begin{abstract}
We investigate the function $L(h,p,q)$, called here the {\em threshold function},
related to periodicity of partial words (words with holes). 
The value $L(h,p,q)$ is defined as the minimum length threshold
which guarantees that a natural extension of the periodicity lemma is valid for partial words with $h$ holes and (strong) periods $p,q$.
We show how to evaluate  the threshold function  in $\Oh(\log p + \log q)$ time,
which is an improvement upon the best previously known $\Oh(p+q)$-time algorithm.
In a series of papers, the formulae for the threshold function, in terms of $p$ and $q$, 
were provided for each fixed $h \le 7$.
We demystify the generic structure of such formulae,
and for each value $h$ we express the threshold function in terms of a piecewise-linear function with $\Oh(h)$ pieces.
\end{abstract}
%
  \section{Introduction}
  Consider a word $X$ of length $|X|=n$, with its positions numbered 0 through $n-1$.
  We say that $X$ has a period $p$ if $X[i]=X[i+p]$ for all $0 \le i < n-p$.
  In this case, the prefix $P=X[0..p-1]$ is called a \emph{string period} of $X$.
  Our work can be seen as a part of the quest to extend Fine and Wilf's Periodicity Lemma \cite{fine1965uniqueness},
  which is a ubiquitous tool of combinatorics on words, into partial words.

  \begin{lemma}[Periodicity Lemma \cite{fine1965uniqueness}]\label{lem:perlem}
    If $p,q$ are periods of a word $X$ of length $|X| \ge p+q-\gcd(p,q)$, then $\gcd(p,q)$ is also a period of $X$.
  \end{lemma}

  A partial word is a word over the alphabet $\Sigma \cup \{\dziura\}$,
  where $\dziura$ denotes a hole (a don't care symbol).
  In what follows, by $n$ we denote the length of the partial word and by $h$ the number of holes.
  For $a,b \in \Sigma \cup \{\dziura\}$, the relation of matching $\approx$
  is defined so that $a\approx b$ if $a=b$ or either of these symbols is a hole.
  A (solid) word $P$ of length $p$ is a \emph{string period} of a partial word $X$
  if $X[i] \approx P[i \bmod p]$ for $0\le i < n$.
  In this case, we say that the integer $p$ is a \emph{(strong) period} of $X$.

  We aim to compute the optimal thresholds $L(h,p,q)$ which make the following generalization of the periodicity lemma valid:
  \begin{lemma}[Periodicity Lemma for Partial Words]\label{lem:perlem2}
     If $X$ is a partial word with $h$ holes with periods $p,q$ and $|X| \ge L(h,p,q)$, then $\gcd(p,q)$ is also a period of $X$.
  \end{lemma}
  If $\gcd(p,q)\in \{p,q\}$, then \cref{lem:perlem2} trivially holds for each partial word $X$. 
  Otherwise, as proved by Fine and Wilf~\cite{fine1965uniqueness}, the threshold in \cref{lem:perlem} is known to be optimal, so $L(0,p,q)=p+q-\gcd(p,q)$.

\begin{example}
$L(1,5,7)=12$, because:
\begin{itemize}
\item each partial word of length at least $12$ with one hole and periods 5, 7
has also period $1=\gcd(5,7)$,
\item the partial word $\al\bl\al\bl\al\al\bl\al\bl\al\dziura$ of length 11 has periods 5, 7 and does not have period $1$.
\end{itemize}
\end{example}
\noindent
As our main aim, we examine the values $L(h,p,q)$ as a function
of $p,q$ for a given $h$. 
Closed-form formulae for $L(h,\cdot,\cdot)$ with $h \le 7$ were given in~\cite{DBLP:journals/tcs/BerstelB99,DBLP:journals/jda/Blanchet-SadriMS12,DBLP:conf/mfcs/ShurK01}.
In these cases, $L(h,p,q)$ can be expressed using a constant number of functions linear in $p$, $q$, and $\gcd(p,q)$.
We discover a common pattern in such formulae which lets us derive a closed-form formula for $L(h,p,q)$
with arbitrary fixed $h$ using a sequence of $\Oh(h)$ fractions. Our construction relies on the theory of continued fractions;
we also apply this link to describe $L(h,p,q)$ in terms of standard Sturmian words. 

As an intermediate step, we consider a dual \emph{holes function} $H(n,p,q)$, which gives the minimum number of holes $h$
for which there is a partial word of length $n$ with $h$ holes and periods $p,q$ which do not satisfy \cref{lem:perlem2}.
\begin{example}
We have $ H(11,5,7)=1$ because:
\begin{itemize}
\item $H(11,5,7)\ge 1$:  due to the classic periodicity lemma, 
every solid word of length 11 with periods $5$ and $7$ has period $1=\gcd(5,7)$, and
\item $H(11,5,7)\le 1$: $\al\bl\al\bl\al\al\bl\al\bl\al\dziura$ is non-unary, has   one
hole and periods 5, 7. 
\end{itemize}

\noindent
We have $H(12,5,7)\le H(11,5,7)+1 =2$ since appending $\dziura$ preserves periods.
In fact $H(12,5,7)=H(15,5,7)=2$. However, there is no non-unary partial word
of length $16$ with 2 holes and periods 5, 7, so $L(2,5,7)=16$; see \cref{TABELKI2}.
\end{example}

\begin{table}[t!]
  \centering
  \setlength\tabcolsep{0.5cm}\begin{tabular}{r|c|p{5cm}}
$h$ & $L(h,5,7)$ & example of length $L(h,5,7)-1$ \\ \hline
0\ & 11\ & $\al\bl\al\bl \al \al\bl\al\bl \al $ \\\hline
1\ & 12\ & $\al\bl\al\bl \al \al\bl\al\bl \al\dziura$ \\\hline
2\ & 16\ & $\al\bl\al\bl \al \al\bl\al\bl \al\dziura\dziura\al\bl\al $  \\\hline
3\ & 19\ & $\al\al\al\al\bl\al\al\al\al\dziura\al\dziura\al\al\dziura\al\al\al $ \\\hline
4\ & 21\ & $\al\bl\al\dziura\dziura\al\bl\al\bl \al \al\bl\al\bl \al\dziura\dziura\al\bl\al$ \\\hline
5\ & 25\ & $\al\al\al\al\bl\al\al\al\al\dziura\al\dziura\al\al\dziura\al\al\al\dziura \dziura \al\al\al\al$
\end{tabular}

  \begin{center}
 \setlength{\tabcolsep}{4pt}
\begin{tabular}{r*{17}{c}}
$n:\ $             & &10&\bf{11}&\bf{12}&13&14&15&\bf{16}&17&18&\bf{19}&20&\bf{21}&22&23&24&\bf{25} \\
\hline
$H(n,5,7):\ $   & & 0& 1& 2& 2& 2& 2 & 3 & 3 & 3 & 4& 4 & 5 & 5& 5& 5& 6
\end{tabular}
\end{center}
  \caption{
The optimal non-unary partial words with periods 5,7 and $h=0,\ldots,5$ holes
(of length $L(h,5,7)-1$)
and the values $H(n,5,7)$ for $n=10,\ldots,25$.
}\label{TABELKI2}
  \end{table}

\noindent For a function $f(n,p,q)$ monotone in $n$, we define its \emph{generalized inverse} as:
$$\widetilde{f}(h,p,q)\;=\; \min\{ n : f(n,p,q) > h\}.$$%
\begin{observation}
$L\;=\; \widetilde{H}$.
\end{observation}

As observed above, \cref{lem:perlem2} becomes trivial if $p \mid q$.
The case of $p\mid 2q$ is known to be special as well, but it has been fully described in~\cite{DBLP:conf/mfcs/ShurK01}.
Furthermore, it was shown in \cite{DBLP:journals/jda/Blanchet-SadriMS12,ShurGamzova04} that the case of $\gcd(p,q)>1$ 
is easily reducible to that of $\gcd(p,q)=1$. 
We recall these existing results in \cref{sec:char}, while in the other sections we assume that $\gcd(p,q)=1$ and $p,q>2$.

\paragraph{Previous results} 
  The study of periods in partial words was initiated by Berstel and Boasson~\cite{DBLP:journals/tcs/BerstelB99},
  who proved that $L(1,p,q)=p+q$. They also showed that the same bound holds for \emph{weak} periods\footnote{%
  An integer $p$ is a weak period of $X$ if $X[i]\approx X[i+p]$ for all $0 \le i < n-p$.} $p$ and $q$.
  Shur and Konovalova~\cite{DBLP:conf/mfcs/ShurK01} developed exact formulae for $L(2,p,q)$ and $L(h,2,q)$,
  and an upper bound for $L(h,p,q)$.
  A formula for $L(h,p,q)$ with small values $h$ was shown by Blanchet-Sadri et al.~\cite{DBLP:journals/iandc/Blanchet-SadriBS08},
  whereas for large $h$, Shur and Gamzova~\cite{ShurGamzova04} proved that the optimal counterexamples
  of length $L(h,p,q)-1$ belong to a very restricted class of \emph{special arrangements}.
  The latter contribution leads to an $\Oh(p+q)$-time algorithm for computing $L(h,p,q)$.
  An alternative procedure with the same running time was shown by Blanchet-Sadri et al.~\cite{DBLP:journals/jda/Blanchet-SadriMS12},
  who also stated closed-form formulae for $L(h,p,q)$ with $h \le 7$.
  Weak periods were further considered in \cite{DBLP:journals/tcs/Blanchet-SadriH02,DBLP:journals/ijfcs/Blanchet-SadriOR10,DBLP:journals/tcs/SmythW09}.

  Other known extensions of the periodicity lemma include a variant with three \cite{DBLP:journals/tcs/CastelliMR99}
  and an arbitrary number of specified periods \cite{DBLP:journals/ita/Justin00,DBLP:journals/tcs/TijdemanZ09},
  the so-called new periodicity lemma \cite{DBLP:journals/dam/BaiFS16,DBLP:journals/siamdm/FanPST06},
  a periodicity lemma for repetitions with morphisms \cite{DBLP:conf/mfcs/ManeaMN12},
  extensions into abelian \cite{DBLP:journals/eatcs/ConstantinescuI06} and $k$-abelian \cite{DBLP:journals/ijfcs/KarhumakiPS13} periodicity,
  into abelian periodicity for partial words \cite{DBLP:journals/ita/Blanchet-SadriSTV13},
  into bidimensional words \cite{DBLP:journals/tcs/MignosiRS03},
  and other variations \cite{DBLP:conf/caap/GiancarloM94,DBLP:conf/mfcs/MignosiSW01}.

\paragraph{Our results} 
First, we show how to compute $L(h,p,q)$ using $\Oh(\log p + \log q)$ arithmetic operations, improving upon the state-of-the-art complexity $\Oh(p+q)$.

Furthermore, for any fixed $h$ in $\Oh(h \log h)$ time we can compute a compact description of the threshold 
function $L(h,p,q)$. For the base case of  $p<q$, $\gcd(p,q)=1$, and  $h<p+q-2$, 
the representation is piecewise linear in $p$ and~$q$.
More precisely, the interval $[0,1]$ can be split into $\Oh(h)$ subintervals $I$ so that $L(h,p,q)$ restricted to $\frac{p}{q}\in I$ is of the form
$a \cdot p + b \cdot q + c$ for some integers $a,b,c$.

\paragraph{Overview of the paper} 
We start by introducing two auxiliary functions $H^s$ and $H^d$ which correspond to two restricted
families of partial words.
Our first key step is to prove that the value $H(n,p,q)$ is always equal to $H^s(n,p,q)$ or $H^d(n,p,q)$
and to characterize the arguments $n$ for which either case holds.
The final function $L$ is then obtained as a combination of the generalized inverses $L^s$ and $L^d$ of $H^s$ and $H^d$, respectively.
Developing the closed-form formula for $L^d$ requires considerable effort;
this is where continued fractions arise.
\section{Functions $H^s$ and $L^s$}\label{sec:HLs}
For relatively prime integers $p,q$, $1<p<q$, and an integer $n\ge q$, let us define
$$H^s(n,p,q) = \floor{\tfrac{n-q}{p}}+\floor{\tfrac{n-q+1}{p}}.$$
We shall prove that $H(n,p,q)\le H^s(n,p,q)$ for a suitable range of lengths $n$.

Fine and Wilf~\cite{fine1965uniqueness} constructed a word of length $p+q-2$ with periods $p$ and $q$ and without period 1.
For given $p,q$ we choose such a word $S_{p,q}$ and,
we define a partial word $W_{p,q}$ as follows, setting $k=\floor{q/p}$ (see~\cref{fig:newfig}):
$$W_{p,q} = (S_{p,q}[0..p-3] \dziura\dziura)^{k}\cdot 
S_{p,q} \cdot(\dziura\dziura S_{p,q}[q..q+p-3])^{k}.$$

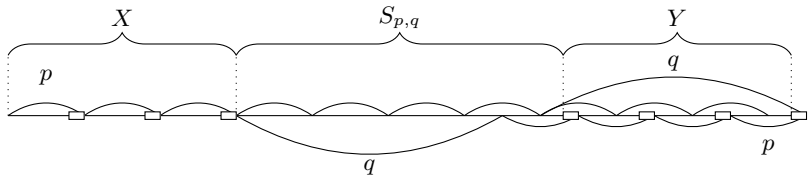
\begin{figure}[b!]
  \begin{center}
\begin{tikzpicture}
\def\unitX{10mm}
\def\unitY{8mm}
\def\unitXX{2mm}
\def\unitP{35mm}
\def\wordXlen{\unitX*3}
\def\wordSlen{\unitX + \unitP - \unitXX}
\def\wordYlen{\unitX*3}
\def\wordLength{\wordXlen+\wordSlen+\wordYlen+\unitXX}
\small

\draw (0,0) -- +(\wordLength,0);
\draw[dotted] (0,0) -- +(0, \unitY);
\draw[dotted] (\wordXlen,0) -- +(0, \unitY);
\draw[dotted] (\wordXlen+\wordSlen,0) -- +(0, \unitY);
\draw[dotted] (\wordXlen+\wordSlen+\wordYlen,0) -- +(0, \unitY);

\draw [decorate,decoration={brace, amplitude=3mm}] (0,\unitY) -- +(\wordXlen,0) node [black,midway, yshift=5mm] {$X$};
\draw [decorate,decoration={brace, amplitude=3mm}] (\wordXlen,\unitY) -- +(\wordSlen,0) node [black,midway, yshift=5mm] {$S_{p,q}$};
\draw [decorate,decoration={brace, amplitude=3mm}] (\wordXlen+\wordSlen,\unitY) -- +(\wordYlen,0) node [black,midway, yshift=5mm] {$Y$};


\foreach \x in {0, 1, 2, 3, 4, 5, 6, 7, 8, 9}{
  \draw (\x*\unitX, 0.0) to [bend left=35] (\x*\unitX+1.0*\unitX,0.0);
}
\foreach \x in {0, 1, 2, 3}{
  \draw (\wordLength-\x*\unitX, 0.0) to [bend left=30] (\wordLength-\x*\unitX-1.0*\unitX,0.0);
}
\node at (\unitX*0.5, 0.3) [above] {$p$};
\node at (\wordLength - \unitX*0.5, -0.2) [below] {$p$};

\foreach \x/\y in {\wordXlen/-30, 7*\unitX/30} {
  \draw (\x, 0) to [bend left=\y] (\x+\unitP, 0);
}
\node at (\wordXlen+0.5*\unitP, -4.75mm) [below] {$q$};
\node at (\wordLength-0.5*\unitP, 4.75mm) [above] {$q$};

\foreach \x in {1, 2, 3} {
  \filldraw[draw=black,fill=white] (\x*\unitX - \unitXX , -0.05) rectangle (\x*\unitX, 0.05);
}
\foreach \x in {0, 1, 2, 3} {
  \filldraw[draw=black,fill=white] (\wordLength - \x*\unitX - \unitXX , -0.05) rectangle (\wordLength - \x*\unitX, 0.05);
}

\end{tikzpicture}
\end{center}
  \caption{The structure of the partial word $W_{p,q}\dziura\dziura\,=\, X\cdot S_{p,q} \cdot Y\dziura\dziura$ for $\floor{q/p}=3$. 
Tiny rectangles correspond to two holes $\dziura\dziura$. We have $|X|=|Y|=p\floor{q/p}=3p$ and
$|W_{p,q}|\;=\; p+q+2p\floor{q/p}-2=q+7p-2$. There are $4\cdot \floor{q/p}=12$ holes.
}\label{fig:newfig}
  \end{figure}

\begin{example}\label{ex:w57}
 For  $p=5$ and $q=7$, we can take $S_{5,7} = \al\bl\al\bl\al\al\bl\al\bl\al$ and
$$W_{5,7}\;=\; \al\bl\al \dziura\dziura\, { \al\bl\al\bl\al\al\bl\al\bl\al}\; \dziura\dziura
\al\bl\al.$$ This partial word has length 20 and 4 holes.
Hence, $H(20,5,7)\le 4=H^s(20,5,7)$ and $L(4,5,7)\ge 21$. In fact, these bounds are tight; see~\cref{TABELKI2}.
\end{example}

Intuitively, the partial word $W_{p,q}$ is an extension of $S_{p,q}$ preserving the period $p$,
in which a small number of symbols is changed to holes to guarantee the periodicity with respect to $q$.
\begin{lemma}\label{lem:claim}
The partial word $W_{p,q}$ has periods $p$ and $q$.
\end{lemma}
\begin{proof}
Let $n=|W_{p,q}|$.
It is easy to observe that $p$ is a period of $W_{p,q}$. 
We now show that $q$ is a period of $W_{p,q}$ as well.
Let $X$ and $Y$ be the prefix and the suffix of $W_{p,q}$ of length $p\floor{q/p}$ (so that $W_{p,q} = X \cdot S_{p,q} \cdot Y$).
Note that $|X|, |Y| < q \le |S_{p,q}|$.

Let us start by showing that $W_{p,q}[i]\approx W_{p,q}[i+q]$ for $0\le i < n-q$. First, suppose that $W_{p,q}[i]$ is contained in $X$. 
The claim is obvious if $i \bmod p \ge p-2$, because in this case we have $W_{p,q}[i]=\dziura$.
Otherwise
$$W_{p,q}[i] = S_{p,q}[i\bmod p] \stackrel{(1)}{=} S_{p,q}[i\bmod p + q] \stackrel{(2)}{=} S_{p,q}[i+q-\lfloor{\tfrac{q}{p}\rfloor}p] = W_{p,q}[i+q],$$
where (1) follows from the fact that $S_{p,q}$ has period $q$ and $i \bmod p < p-2$, and (2) from the fact that $S_{p,q}$ has period $p$.
By symmetry of our construction, we also have $W_{p,q}[i]\approx W_{p,q}[i+q]$ if $W_{p,q}[i+q]$ is contained in $Y$.
In the remaining case, $W_{p,q}[i]$ and $W_{p,q}[i+q]$ are both contained in $S_{p,q}$, which yields $W_{p,q}[i+q]=W_{p,q}[i]$.

Next, we claim that $W_{p,q}[i]\approx W_{p,q}[i+kq]$ for every $k \ge 2$.
Observe that $W_{p,q}[i+q],\ldots, W_{p,q}[i+(k-1)q]$ are contained in $S_{p,q}$ and thus they are equal solid symbols.
Hence, $$W_{p,q}[i]\approx W_{p,q}[i+q] =\cdots = W_{p,q}[i+(k-1)q]\approx W_{p,q}[i+kq].$$
 The intermediate symbols are solid, so this implies $W_{p,q}[i]\approx W_{p,q}[i+kq]$,
as claimed.
Consequently, $q$ is indeed a period of $W_{p,q}$.
\end{proof}

We use the word $S_{p,q}$ and the partial word $W_{p,q}\dziura\dziura$  to show that $H^s$ is an upper bound for $H$
for all intermediate lengths $n$ ($|S_{p,q}|\le n \le |W_{p,q}\dziura\dziura|$).

\begin{lemma}\label{lem:hs}
Let $1<p<q$ be relatively prime integers.
For each length $p+q-2 \le n \le p+ q + 2p\floor{q/p}$,  we have $H(n,p,q)\le H^s(n,p,q)$.
\end{lemma}

\begin{proof}
We extend $S_{p,q}$ to $W_{p,q}\dziura\dziura$ symbol by symbol, first
prepending the characters before $S_{p,q}$, and then appending the characters
after $S_{p,q}$. By \cref{lem:claim}, the resulting partial word has periods $p$
and $q$ because it is contained in $W_{p,q}\dziura\dziura$. Moreover, it is
not unary because it contains $S_{p,q}$.

A hole is added at the first two iterations among every $p$ iterations. Hence, the total number of holes is as claimed:
$$\ceil{\tfrac{n-|S_{p,q}|}{p}} + \ceil{\tfrac{n-|S_{p,q}|-1}{p}} = \floor{\tfrac{n-q+1}{p}} + \floor{\tfrac{n-q}{p}} = H^s(n,p,q),$$
because $\lceil{\frac{x}{p}\rceil}=\lfloor{\frac{x+p-1}{p}}\rfloor$ for every
integer $x$.

\end{proof}

Finally, the function $L^s=\widetilde{H^s}$ is very simple and easily computable.
\begin{lemma}\label{lem:ls}
If $h\ge 0$ is an integer, then
  $L^s(h,p,q) = \ceil{\frac{h+1}{2}} p + q - (h+1)\bmod 2$.
\end{lemma}
\begin{proof}
We have to determine the smallest $n$ such that 
$\floor{\frac{n-q}{p}}+\floor{\frac{n-q+1}{p}}\;=\;h+1.$
There are two cases, depending on parity of $h$:
\paragraph{Case 1: $h=2k$.} In this case $\floor{\frac{n-q}{p}}=k$
and $\floor{\frac{n-q+1}{p}}=k+1$.
Hence, $n-q+1=p(k+1)$, i.e., $n=p(k+1)+q-1=\ceil{\frac{h+1}{2}} p + q - (h+1)\bmod 2$. 
\paragraph{Case 2: $h=2k+1$.} In this case $\floor{\frac{n-q}{p}}=k+1$
and $\floor{\frac{n-q+1}{p}}=k+1$. Hence, $n-q=p(k+1)$,
i.e., $n=p(k+1)+q=\ceil{\frac{h+1}{2}} p + q - (h+1)\bmod 2$.
 \end{proof}

\section{Functions $H^d$ and $L^d$}\label{sec:HLd}
In this section, we study a family of partial words corresponding to the \emph{special arrangements} introduced in \cite{ShurGamzova04}.
For relatively prime integers $p,q>1$, we say that a partial word $S$ of length $n\ge \max(p,q)$ is \emph{$(p,q)$-special}
if it has a position $l$ such that for each position $i$:
 $$S[i]=\begin{cases}
  \al & \text{if }p \nmid (l-i)\text{ and }q\nmid (l-i),\\
 \bl &\text{if }p \mid (l-i)\text{ and }q\mid (l-i),\\
 \dziura&\text{otherwise.}
 \end{cases}$$
Let $H^d(n,p,q)$ be the minimum number of holes in a $(p,q)$-special partial word of length~$n$.%
\begin{fact}\label{obs:Hdlb}
For each $n\ge \max(p,q)$, we have $H(n,p,q)\le H^d(n,p,q)$.
\end{fact}
\begin{proof}
Observe that every $(p,q)$-special partial word has periods $p$ and $q$.
However, due to $p,q>1$, it does not have period $1=\gcd(p,q)$.
\end{proof}
\begin{example}
The partial word $\al\al\al\al\bl\al\al\al\al \dziura \al\dziura \al\al \dziura \al\al\al$ is $(5,7)$-special (with $l=4$), so $H(18,5,7) \le H^d(18,5,7)\le 3$
and $L(3,5,7)\ge 19$. In fact, these bounds are tight; see~\cref{TABELKI2}.
\end{example}

To derive a formula for $H^d(n,p,q)$, let us introduce an auxiliary function $G$,
which counts integers $i\in \{1,\ldots,n\}$ that are multiples of $p$ or of $q$ but not both:
$$G(n,p,q)=\floor{\tfrac{n}{p}}+\floor{\tfrac{n}{q}}-2\floor{\tfrac{n}{pq}}.$$
The function $H^d$ can be characterized using $G$, while 
the generalized inverse $L^d = \widetilde{H^d}$ admits a dual characterization in terms of $\widetilde{G}$; see also \cref{tab:gld}.
\begin{lemma}\label{lem:simple_hd}
Let $p,q>1$ be relatively prime integers.
\begin{enumerate}[label={\rm(\alph*)}]
  \item If $n \ge \max(p,q)$, then
$H^d(n,p,q)=\min_{l=0}^{n-1}\left(G(l,p,q)+G(n-l-1,p,q)\right)$.
\item If $h\ge 0$, then
  $L^d(h,p,q)=\max_{k=0}^h\left(\widetilde{G}(k,p,q)+\widetilde{G}(h-k,p,q)\right).$
\end{enumerate}
\end{lemma}
\begin{proof}
Let $S$ be a $(p,q)$-special partial word of length $n$ with $h$ holes,
$k$ of which are located to the left of position $l$.
Observe that $k = G(l,p,q)$ (so $l+1\le \widetilde{G}(k,p,q)$)
and $h-k = G(n-l-1,p,q)$ (so $n-l \le \widetilde{G}(h-k,p,q)$).
Hence,
$h = G(l,p,q)+G(n-l-1,p,q)$ and $n+1 \le \widetilde{G}(k,p,q)+\widetilde{G}(h-k,p,q)$.
The claimed equalities follow from the fact that these bounds can be attained for
each $l$ and $k$, respectively.
\end{proof}

  \begin{table}[t!]
  \centering
\renewcommand*{\arraystretch}{1.25}
\setlength{\tabcolsep}{4pt}
\begin{tabular}{r|*{22}{c}}
$h\ $             &0&1&2&3&4&5&6&7&8&9&10&11&12&13&14&15&16&17&18&19&20 \\
\hline
$\widetilde{G}(h,5,7)\ $   & 5 & 7 & 10 & 14 & 15 & 20 & 21 & 25 & 28 & 30 & 40 & 42 & 45 & 49 & 50 & 55 & 56 & 60 & 63 & 65 & 75\\
\hline
$L^d(h,5,7)\ $   & 10 & 12 & 15 & 19 & 21 & 25 & 28 & 30 & 34 & 35 & 45 & 47 & 50 & 54 & 56 & 60 & 63 & 65 & 69 & 70 & 80
\end{tabular}
\caption{Functions $\widetilde{G}$ and $L^d$ for $p=5$, $q=7$, and $h=0,\ldots,20$.
By \cref{lem:simple_hd}, we have, for example,  $L^d(8,5,7)=\max\big(\widetilde{G}(0,5,7)+\widetilde{G}(8,5,7),\ldots,\widetilde{G}(4,5,7)+\widetilde{G}(4,5,7)\big)
              = \max({5+28},\,7+25,\,10+21,\,14+20,\,15+15)=34$.
              }\label{tab:gld}
\end{table}

\section{Characterizations of $H$ and $L$}\label{sec:char}
Shur and Gamzova in \cite{ShurGamzova04} proved that $H(n,p,q)\,=\, H^d(n,p,q)$ for $n\ge 3q+p$.
In this section, we give a complete characterization of $H$ in terms of $H^d$ and $H^s$,
and we derive an analogous characterization of $L$ in terms of $L^d$ and $L^s$.
Our proof is based on a graph-theoretic approach similar to that in~\cite{DBLP:journals/jda/Blanchet-SadriMS12}.

Let us define the \emph{$(n,p,q)$-graph} $\mathbf{G}=(V,E)$ as an undirected graph with vertices $V=\{0,\ldots,n-1\}$.
The vertices $i$ and $j$ are connected if and only if $p\mid (j-i)$ or $q \mid (j-i)$.
Observe that $H(n,p,q)$ is the minimum size of a \emph{vertex separator} in $\mathbf{G}$,
i.e., the minimum number of vertices to be removed from $\mathbf{G}$ so that the resulting graph is no longer connected; see~\cref{fig:kratka}.

We say that an edge $(i,j)$ of the $(n,p,q)$-graph is a \emph{$p$-edge} if $p \mid (j-i)$ and a \emph{$q$-edge} if $q \mid (j-i)$.
The set of all nodes giving the same remainder modulo $p$ (modulo $q$) is called a \emph{$p$-class} (\emph{$q$-class}, respectively).
Each $p$-class and each $q$-class forms a clique in the $(n,p,q)$-graph.

\begin{fact}[see  \cite{DBLP:journals/jda/Blanchet-SadriMS12}]
Let $1<p<q$ be relatively prime integers.
If $n<pq$, then $H^d(n,p,q)$ is the minimal degree of a vertex in the $(n,p,q)$-graph.
\label{lem:blanchet}
\end{fact}
\begin{proof}
Observe that vertex number $l$ has $G(l,p,q)$ neighbors $i<l$ and $G(n-l-1,p,q)$ neighbors $i>l$. 
Consequently, by \cref{lem:simple_hd}, $H^d(n,p,q)=\min_{l=0}^{n-1}(G(l,p,q)+G(n-l-1,p,q))=\min_{l=0}^{n-1}\deg_{\mathbf{G}}(l)$.
\end{proof}

\noindent
Let $\mathbf{G}=(V,E)$ be the $(n,p,q)$-graph.
For each $i\in \{0,\ldots,p-1\}$ let $C_i$ be the $p$-class containing the vertex $i$; see \cref{fig:kratka}.
We slightly abuse the notation and use arbitrary integers for indexing the $p$-classes: $C_i = C_{i \bmod p}$ for $i\in \mathbb{Z}$.
We denote by $E_i$ the set of $q$-edges of the form $(j,j+q)$ for $j\in C_i$.
Let us start with two auxiliary facts.

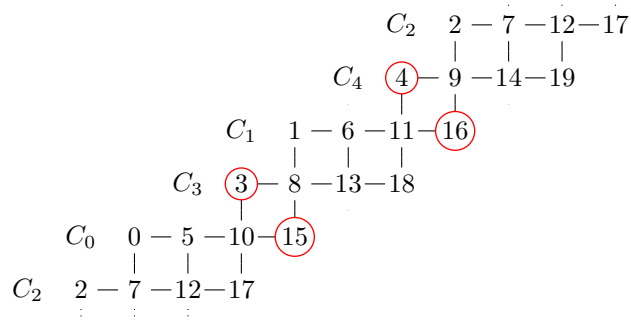
\begin{figure}[htpb]
  \begin{center}
  \begin{tikzpicture}
\def\unit{20pt}
\tikzstyle{empty}=[draw, circle, minimum size=\unit*0.6, inner sep=0pt, fill=white, line width=0pt, draw=white]
\tikzstyle{marked}=[draw, circle, minimum size=\unit*0.6, inner sep=1pt, fill=white, line width=0.5pt, draw=red]

\newcommand{\drawGrid}[3]{
  \draw (1*\unit,0)--(#1*\unit,0);
  \foreach \x in {2,...,#1} {
    \draw[#2] (\x*\unit,0)--+(0,0.5*\unit);
    \draw[#3] (\x*\unit-\unit,0)--+(0,-0.5*\unit);
  }
}

\begin{scope}
  \drawGrid{4}{}{dotted}
  \node at (0,0) {$C_2$};
  \foreach \label/\style [count=\i] in {2/empty, 7/empty, 12/empty, 17/empty} {
    \node [\style] at (\i*\unit, 0) {$\label$};
  }
\end{scope}

\begin{scope}[xshift=1*\unit,yshift=1*\unit]
  \drawGrid{4}{}{}\draw[white,line width=2pt](2*\unit,0)--+(0,\unit);
  \node at (0,0) {$C_0$};
  \foreach \label/\style [count=\i] in {0/empty, 5/empty, 10/empty, 15/marked} {
    \node [\style] at (\i*\unit, 0) {$\label$};
  }
\end{scope}

\begin{scope}[xshift=3*\unit,yshift=2*\unit]
  \drawGrid{4}{}{}\draw[white,line width=2pt](3*\unit,0)--+(0,-0.5*\unit);
  \node at (0,0) {$C_3$};
  \foreach \label/\style [count=\i] in {3/marked, 8/empty, 13/empty, 18/empty} {
    \node [\style] at (\i*\unit, 0) {$\label$};
  }
\end{scope}

\begin{scope}[xshift=4*\unit,yshift=3*\unit]
  \drawGrid{4}{}{}\draw[white,line width=2pt](2*\unit,0)--+(0,0.5*\unit);
  \node at (0,0) {$C_1$};
  \foreach \label/\style [count=\i] in {1/empty, 6/empty, 11/empty, 16/marked} {
    \node [\style] at (\i*\unit, 0) {$\label$};
  }
\end{scope}

\begin{scope}[xshift=6*\unit,yshift=4*\unit]
  \drawGrid{4}{}{}\draw[white,line width=2pt](3*\unit,0)--+(0,-0.5*\unit);
  \node at (0,0) {$C_4$};
  \foreach \label/\style [count=\i] in {4/marked, 9/empty, 14/empty, 19/empty} {
    \node [\style] at (\i*\unit, 0) {$\label$};
  }
\end{scope}

\begin{scope}[xshift=7*\unit,yshift=5*\unit]
  \drawGrid{4}{dotted}{}
  \node at (0,0) {$C_2$};
  \foreach \label/\style [count=\i] in {2/empty, 7/empty, 12/empty, 17/empty} {
    \node [\style] at (\i*\unit, 0) {$\label$};
  }
\end{scope}

\end{tikzpicture}
  \end{center}
  \caption{The structure of the $(20,5,7)$-graph. 
Each 5-clique $C_i$ is actually a clique; same applies for the vertical $7$-cliques.
The $5$-clique $C_2$ is repeated to show the cyclicity.
The set $U=\{3,4,5,7\}$ of encircled vertices is a minimum-size vertex separator.
It corresponds to the partial word $W_{5,7}$ from \cref{ex:w57}:
holes of $W_{5,7}$ are located at positions $i\in U$,
the positions $i$ with $W_{5,7}[i]=\bl$ form a connected component $(C_1\cup C_4)\setminus U$,
while the positions $i$ with $W_{5,7}[i]=\al$ form a connected component $(C_0\cup C_2\cup C_3)\setminus U$.
}\label{fig:kratka}
  \end{figure}

\begin{fact}\label{fct:HsE}
Let $1<p<q$ be relatively prime integers.
\begin{enumerate}[label={\rm(\alph*)}]
  \item For $j\in \{0,\ldots,p-1\}$, we have $|E_j| = \ceil{\frac{n-j-q}{p}}$.
  \item $H^s(n,p,q) = |E_{p-1}|+|E_{p-2}|=\min_{i \ne j} \left(|E_i| + |E_j|\right)$.
 \end{enumerate}
\end{fact}
 \newpage
\begin{proof}
  Let $i=kp+j$, where $0 \le j < p$.
  There is a $q$-edge $(i,i+q)$ if and only if
  $$kp+j+q \le n-1,\text { so }k \le \floor{\tfrac{n-1-j-q}{p}}.$$
  The number of such values of $k$ is $\floor{\tfrac{n-1-j-q}{p}}+1=\ceil{\tfrac{n-j-q}{p}}$.

  As for the second statement of the fact, we have:
  $$|E_j| \ge |E_{p-1}| = \ceil{\tfrac{n-p+1-q}{p}}=\floor{\tfrac{n-q}{p}}$$
  for $0\le j < p$
  and, similarly, $|E_j|\ge |E_{p-2}|=\floor{\tfrac{n-q+1}{p}}$ for $0\le j < p-1$.
\end{proof}

\begin{fact}\label{fct:color}
Let $U$ be a vertex separator in the $(n,p,q)$-graph $\mathbf{G}=(V,E)$
and let $\mathbf{G}'=\mathbf{G} \setminus U$.
  One can color the vertices of $\mathbf{G}$ in two colors so that
  every edge in $\mathbf{G}'$ and every $p$-class in $\mathbf{G}$ is monochromatic, 
  but $\mathbf{G}'$ is not monochromatic.
\end{fact}
\begin{proof}
Recall that each $p$-class $C_i$ is a clique in $\mathbf{G}$, so $C_i \setminus U$ is still a clique in $\mathbf{G}'$.
We distinguish a connected component $M$ of $\mathbf{G}'$ and color the vertices of $C_i$ depending on whether $C_i\setminus U \subseteq M$.
It is easy to verify that this coloring satisfies the claimed conditions.
 \end{proof}

The following lemma provides lower bounds on $H(n,p,q)$.

\begin{lemma}\label{lem:main}
Let $1<p<q$ be relatively prime integers.
  \begin{enumerate}[label={\rm(\arabic*)}]
    \item If $n < 2q$, then $H(n,p,q) \ge H^s(n,p,q)$.
    \item If $p\ge 3$ and $n \ge 2q$, then $H(n,p,q) \ge \min(H^d(n,p,q), H^s(n,p,q))$.
    \item If $p \ge 5$ and $n \ge 4q$, then $H(n,p,q) \ge \min(H^d(n,p,q), H^s(n,p,q)+1)$.
  \end{enumerate}
\end{lemma}
\begin{proof}
Let $U$ be a minimum-size vertex separator of hte $(n,p,q)$-graph  $\mathbf{G}=(V,E)$;
recall that $|U|=H(n,p,q)$.
  Let us fix a coloring of $\mathbf{G}$ using colors $\{A,B\}$ satisfying \cref{fct:color}; without loss of generality we assume that the number of $p$-classes
  with color $A$ is at least the number of $p$-classes with color $B$.
  We have the following two cases.

  \paragraph{Case a: Exactly one $p$-class has color $B$.}
  Let $C_j$ be the unique $p$-class with color $B$.
  By definition, the edges in $E_{j-q}\cup E_j$ are bichromatic. 
  If $n<2q$, then all the $q$-edges form a matching in $\mathbf{G}$.
  In particular, in order to disconnect $\mathbf{G}$, we need to remove at least one endpoint of each edge in $E_{j-q}\cup E_j$.
  Hence, $H(n,p,q) \ge |E_{j-q}|+|E_j|\ge H^s(n,p,q)$, where the second inequality follows from \cref{fct:HsE}(b).
  This concludes the proof of (1) in this case.

  Now assume that $n \ge 2q$ and $p\ge 3$.
  We will show that $H(n,p,q) \ge H^d(n,p,q)$ holds in this case.
  Consider any $q$-class $D$ and let $k$ be its size; we have $k \ge \floor{\frac{n}{q}}\ge 2$.
  In this $q$-class, every $p$-th element has color $B$.
  Let $\#_A(D)$ and $\#_B(D)$ denote the number of vertices in $D$ colored with $A$ and $B$, respectively.
  Then:
  $$\#_B(D) \le \ceil{\tfrac{k}{p}} \le \ceil{\tfrac{k}{3}} = \floor{\tfrac{k+2}{3}} \le \floor{\tfrac{2k}{3}} = k - \ceil{\tfrac{k}{3}} \le k - \ceil{\tfrac{k}{p}} \le \#_A(D).$$
  The set $U$ contains all $B$-colored vertices or all $A$-colored vertices of every $q$-class $D$, as otherwise there would be a non-monochromatic edge in $\mathbf{G}'$
  connecting two vertices of $D \setminus U$, contradicting \cref{fct:color}.
  At least one vertex of $\mathbf{G}'$ is $B$-colored, so in at least one $q$-class, $U$ must contain all $A$-colored vertices; assume that this is the $q$-class $D_0$.
  Consequently,
  \begin{multline*}
    |U| = \sum_{D: q\text{-class}} |U \cap D| \ge \#_A(D_0) + \sum_{D: q\text{-class},\, D\ne D_0} \#_B(D) =\\
     |D_0| +  \sum_{D: q\text{-class}} \#_B(D)  - 2\#_B(D_0) = |D_0| + |C_j| - 2|D_0 \cap C_j| \ge H^d(n,p,q).
  \end{multline*}
  The last inequality follows from \cref{lem:blanchet}.
  This concludes (2) and (3) in this case.

  \paragraph{Case b: There are at least two $p$-classes with each color.}
  In particular, $p \ge 4$.
  We consider two subcases based on the colors $c_i$ of classes $C_i$.
  In each case we will show that $H(n,p,q)$ is bounded from below by $H^s(n,p,q)$ or $H^s(n,p,q)+1$.
  
  First, suppose that there is exactly one $p$-class $C_i$ such that $c_i = A$ and $c_{i+q}=B$.
  Equivalently, there is exactly one $p$-class $C_j$ such that $c_j = B$ and $c_{j+q}=A$.
  Since there are at least two $p$-classes with each color, $c_{i+2q}=B$ and $c_{j+2q}=A$, so $C_{i+q}\ne C_j$ and $C_{j+q}\ne C_i$.
  This means that $E_i \cup E_j$ forms a bichromatic matching in $\mathbf{G}$.
  Consequently, 
    $$H(n,p,q) \ge |E_{i}| + |E_{j}| \ge H^s(n,p,q).$$
  This concludes the proof of (1) and (2) in the current subcase.
  
  For the proof of (3), observe that $p\ge 5$ and the choice of $A$ as the more frequent color yields
  that $c_{j+3q}=A$, so $C_{j+2q}$ is distinct from $C_i$.
  Hence, we can extend the matching $E_i \cup E_j$ with an edge $(x,y)$
  where $x=(j-q)\bmod p \in C_{j-q}$ and $y = x+3q \in C_{j+2q}$.
  This edge exists because $n \ge 4q > p+3q > y$. It forms a matching with $E_i\cup E_j$ because
  no edge in $E_i\cup E_j$ is incident to $C_{j+2q}$, while the only edges incident to $C_{j-q}$
  could be the edges in $E_i$ provided that $C_{i}=C_{j-2q}$. However, $x < q$, so $x$
  is not an endpoint of any edge in $E_{j-2q}$.
  This concludes the proof of (3) in the current subcase.
 
  Let us proceed to the second subcase.
  Let $C_i,C_j$ be two distinct $p$-classes such that $c_i=c_j = A$ and $c_{i+q}=c_{j+q} = B$.
  It is easy to see that $E_i \cup E_j$ forms a bichromatic matching in $\mathbf{G}$, so $H(n,p,q) \ge |E_{i}| + |E_{j}| \ge H^s(n,p,q)$.
  Thus, it remains to prove (3) in this subcase.
  
  If $p\ge 5$, then there is a third $p$-class $C_k$ with color $A$.
  Moreover, we may choose $k$ so that $c_{k-q}=B$ and $c_{k}=A$.
  We extend $E_i\cup E_j$ with an edge $(x,y)$ where $x=(k-q)\bmod p \in C_{k-q}$ and $y = x+q \in C_k$.
  This edge exists because $n > q+p > y$. It forms a matching with $E_i\cup E_j$ because
  no edge in $E_i \cup E_j$ is incident to $C_{k}$, while the only edges incident to $C_{k-q}$
  might be the edges in $E_i$ or $E_j$ provided that $i=k-2q$ or $j=k-2q$. However, $x<q$, so $x$
  is not an endpoint of any edge in $E_{k-2q}$.
  This concludes the proof of (3) in the current subcase, and the proof of the entire lemma.   
 \end{proof}

\begin{lemma}\label{lem:hdb}
If $n \ge q$, then
$$ \floor{\tfrac{n}{q}}+\floor{\tfrac{n}{p}} -2\ceil{\tfrac{n}{pq}} \le H^d(n,p,q) \le \floor{\tfrac{n}{q}}+\floor{\tfrac{n-q}{p}} + \floor{\tfrac{q-1}{p}} - 1.$$
\end{lemma}
\begin{proof}
Recall that $H^d(n,p,q)=\min_{l=0}^{n-1} \left(G(l,p,q)+G(n-l-1,p,q)\right)$ due to \cref{lem:simple_hd}.
The first part of the claim holds because for $0\le l < n$ we have:
\begin{multline*}G(l,p,q)+G(n-l-1,p,q)=\floor{\tfrac{l}{p}}+\floor{\tfrac{l}{q}}-2\floor{\tfrac{l}{pq}}+\floor{\tfrac{n-l-1}{p}}+\floor{\tfrac{n-l-1}{q}}-2\floor{\tfrac{n-l-1}{pq}}=\\
\ceil{\tfrac{l-p+1}{p}}+\ceil{\tfrac{l-q+1}{q}}-2\floor{\tfrac{l}{pq}}+\ceil{\tfrac{n-l-p}{p}}+\floor{\tfrac{n-l-q}{q}}-2\floor{\tfrac{n-l-1}{pq}}\ge \\
\ceil{\tfrac{n-2p+1}{p}}+\ceil{\tfrac{n-2q+1}{q}}-2\floor{\tfrac{n-1}{pq}}=
\floor{\tfrac{n}{p}}-1+\floor{\tfrac{n}{q}}-1-2\ceil{\tfrac{n}{pq}}+2=\floor{\tfrac{n}{q}}+\floor{\tfrac{n}{p}} -2\ceil{\tfrac{n}{pq}}.\end{multline*}
As for the second part, due to $n\ge q$ we have:
\begin{multline*}
H^d(n,p,q)\ge G(q-1,p,q)+G(n-q,p,q)=
\floor{\tfrac{q-1}{p}}+\floor{\tfrac{n-q}{p}}+\floor{\tfrac{n-q}{q}}-2\floor{\tfrac{n-q}{pq}}\le \\
\floor{\tfrac{q-1}{p}}+\floor{\tfrac{n-q}{p}}+\floor{\tfrac{n}{q}}-1,
\end{multline*}
which completes he proof.
\end{proof}

\begin{theorem}\label{thm:wr}
Let $p$ and $q$ be relatively prime integers such that $2<p<q$.
For each integer $n\ge p+q-2$, we have
$$
H(n,p,q)=
\begin{cases}
H^s(n,p,q) & \text{if }n \le  q +p\lceil{\frac{q}{p}}\rceil-1 \mbox{ or } 3q \le n \le q+3p-1, \\
H^d(n,p,q) & \text{otherwise. }
\end{cases}
$$
Moreover, for each integer $h \ge 0$:
$$
L(h,p,q) =
\begin{cases}
L^s(h,p,q) & \text{if } \frac{q}{p} > \ceil{\frac{h}{2}}\text{ or } (h=4 \mbox{ and }\frac{q}{p} < \frac{3}{2})\\
L^d(h,p,q) & \text{otherwise.}
\end{cases}
$$
\end{theorem}

\begin{proof}
First, we prove the claim concerning $H$ by analyzing several cases.
\subsubsection*{Case 0.\ $pq \le n$.}
By \cref{obs:Hdlb}, we have $H(n,p,q)\le H^d(n,p,q)$.
Moreover, using \cref{lem:hdb}, we obtain:
\begin{multline*}
H^s(n,p,q) = \floor{\tfrac{n-q}{p}}+\floor{\tfrac{n-q+1}{p}} \ge \floor{\tfrac{n-q}{p}}+\floor{\tfrac{n-pq}{p}}+\floor{\tfrac{q-1}{p}}+\floor{\tfrac{q(p-2)+2}{p}}\ge\\
\ge \floor{\tfrac{n-q}{p}}+\floor{\tfrac{n}{q}}-p+\floor{\tfrac{q-1}{p}}+\floor{\tfrac{(p+1)(p-2)+2}{p}}=\floor{\tfrac{n-q}{p}}+\floor{\tfrac{n}{q}}+\floor{\tfrac{q-1}{p}}+\floor{\tfrac{p^2-p-p^2}{p}}=\\
=\floor{\tfrac{n-q}{p}}+\floor{\tfrac{n}{q}}+\floor{\tfrac{q-1}{p}}-1 \ge H^d(n,p,q).
\end{multline*}
Finally, \cref{lem:main}(2) yields that $H(n,p,q)\ge \min(H^d(n,p,q), H^s(n,p,q)) = H^d(n,p,q)$,
which completes the proof.

Henceforth we assume that $n<pq$.

\subsubsection*{Case 1. $p+q-1 \le n < 2q$.}
We get $H(n,p,q)=H^s(n,p,q)$ directly from \cref{lem:main}(1) and \cref{lem:hs}.

\subsubsection*{Case 2. $2q \le n \le q+\ceil{\frac{q}{p}}p-1$.}
Note that $n \le q+\ceil{\frac{q}{p}}p-1 = p+q+\floor{\frac{q}{p}}p-1 < p+q+2p\floor{\frac{q}{p}}$,
so $H(n,p,q)\le H^s(n,p,q)$ due to \cref{lem:hs}.
Moreover, using \cref{lem:hdb}, we obtain:
\begin{multline*}
H^s(n,p,q) = \floor{\tfrac{n-q}{p}}+\floor{\tfrac{n-q+1}{p}}\le \floor{\tfrac{\ceil{\frac{q}{p}}p-1}{p}}+ \floor{\tfrac{n-q+1}{p}}=\ceil{\tfrac{q}{p}}-1+ \floor{\tfrac{n-q+1}{p}} =\floor{\tfrac{q-1}{p}}+\floor{\tfrac{n-q+1}{p}} \le \\
\floor{\tfrac{n}{p}}\le
 \floor{\tfrac{n}{p}}+\floor{\tfrac{n}{q}}-2 \le H^d(n,p,q).
\end{multline*}
Finally, \cref{lem:main}(2) yields that $H(n,p,q)\ge \min(H^d(n,p,q), H^s(n,p,q)) = H^s(n,p,q)$,
which completes the proof.

\subsubsection*{Case 3. $q+\ceil{\frac{q}{p}}p-1\le n < 3q$.}
By \cref{obs:Hdlb}, we have $H(n,p,q)\le H^d(n,p,q)$.
Moreover, using \cref{lem:hdb}, we obtain:
\begin{multline*}
H^s(n,p,q) = \floor{\tfrac{n-q}{p}}+\floor{\tfrac{n-q+1}{p}} \ge \floor{\tfrac{n-q}{p}}+\floor{\tfrac{\ceil{\tfrac{q}{p}}p}{p}}= \floor{\tfrac{n-q}{p}}+\ceil{\tfrac{q}{p}} =\\
\floor{\tfrac{n-q}{p}}+\floor{\tfrac{q-1}{p}}+1+\floor{\tfrac{n}{q}}-2 \ge H^d(n,p,q).
\end{multline*}
Finally, \cref{lem:main}(2) yields that $H(n,p,q)\ge \min(H^d(n,p,q), H^s(n,p,q)) = H^d(n,p,q)$,
which completes the proof.

\subsubsection*{Case 4. $3q\le n \le 3p+q-1$.}
Note that $n \le 3p+q-1 < p+q + 2p \le p+q+2p\floor{\frac{q}{p}}$,
so $H(n,p,q)\le H^s(n,p,q)$ due to \cref{lem:hs}.
Moreover, using \cref{lem:hdb}, we obtain:
\begin{multline*}
H^s(n,p,q) = \floor{\tfrac{n-q}{p}}+\floor{\tfrac{n-q+1}{p}}\le \floor{\tfrac{3p-1}{p}}+\floor{\tfrac{n-q+1}{p}} = 2+\floor{\tfrac{n-q+1}{p}} \le 1+\floor{\tfrac{q-1}{p}}+\floor{\tfrac{n-q+1}{p}}\le \\
\le 1+\floor{\tfrac{n}{p}} = 3 +\floor{\tfrac{n}{p}}-2 \le \floor{\tfrac{n}{q}}+\floor{\tfrac{n}{p}}-2 \le H^d(n,p,q).
\end{multline*}
Finally, \cref{lem:main}(2) yields that $H(n,p,q)\ge \min(H^d(n,p,q), H^s(n,p,q)) = H^s(n,p,q)$,
which completes the proof.

\subsubsection*{Case 5. $\max(3q,3p+q-1)\le n < 4q$ and $p<q < 2p$.}
By \cref{obs:Hdlb}, we have $H(n,p,q)\le H^d(n,p,q)$.
Moreover, using \cref{lem:hdb}, we obtain:
\begin{multline*}
H^s(n,p,q)\!=\! \floor{\tfrac{n-q}{p}}+\floor{\tfrac{n-q+1}{p}} \ge \floor{\tfrac{n-q}{p}}+\floor{\tfrac{3p}{p}}= \floor{\tfrac{n-q}{p}}+3  = \floor{\tfrac{n-q}{p}}+\floor{\tfrac{q-1}{p}}+\floor{\tfrac{n}{q}}-1 \ge H^d(n,p,q).
\end{multline*}
Finally, \cref{lem:main}(2) yields that $H(n,p,q)\ge \min(H^d(n,p,q), H^s(n,p,q)) = H^d(n,p,q)$,
which completes the proof.

\subsubsection*{Case 6. $3q \le n$ and $q > 2p$.}
By \cref{obs:Hdlb}, we have $H(n,p,q)\le H^d(n,p,q)$.
Moreover, using \cref{lem:hdb}, we obtain:
\begin{multline*}
H^s(n,p,q) \ge 2\floor{\tfrac{n-q}{p}} \ge \floor{\tfrac{n-q}{p}}+\floor{\tfrac{n-3q}{p}}+2\floor{\tfrac{q}{p}}
\ge \floor{\tfrac{n-q}{p}}+\floor{\tfrac{n-3q}{q}}+2\floor{\tfrac{q}{p}} \ge\\\ge \floor{\tfrac{n-q}{p}}+\floor{\tfrac{n}{q}}-3+\floor{\tfrac{q-1}{p}}+2= \floor{\tfrac{n-q}{p}}+\floor{\tfrac{q-1}{p}}+\floor{\tfrac{n}{q}}-1 \ge H^d(n,p,q).
\end{multline*}
Finally, \cref{lem:main}(2) yields that $H(n,p,q)\ge \min(H^d(n,p,q), H^s(n,p,q)) = H^d(n,p,q)$,
which completes the proof.

\subsubsection*{Case 7.\ $4q \le n$ and $p \ge 5$ (and $q < 2p$).}
By \cref{obs:Hdlb}, we have $H(n,p,q)\le H^d(n,p,q)$.
Moreover, using \cref{lem:hdb}, we obtain:
\begin{multline*}
H^s(n,p,q) =\floor{\tfrac{n-q}{p}}+\floor{\tfrac{n+1-q}{p}} \ge \floor{\tfrac{n-q}{p}}+\floor{\tfrac{n-2q}{p}}+\floor{\tfrac{q-1}{p}}\ge \floor{\tfrac{n-q}{p}}+\floor{\tfrac{n-2q}{q}}+\floor{\tfrac{q-1}{p}}=\\
=\floor{\tfrac{n-q}{p}}+\floor{\tfrac{n}{q}}-2+\floor{\tfrac{q-1}{p}} \ge H^d(n,p,q)-1.
\end{multline*}
Finally, \cref{lem:main}(3) yields that $H(n,p,q)\ge \min(H^d(n,p,q), H^s(n,p,q)+1) = H^d(n,p,q)$,
which completes the proof.

\medskip
The only remaining case, that $4q \le n$ and $p<5$, is a subcase of Case 0.
This completes the proof of the formula for $H(n,p,q)$.

\newpage

The characterization of $L(h,p,q)$ is relatively easy to derive from that of $H(n,p,q)$.
 Recall that $L$, $L^s$, and $L^d$ are generalized inverses of $H$, $H^s$, and $H^d$, respectively.
  Note that Cases 1.\ and 2.\ yield $H(n,p,q)=H^s(n,p,q)$ for $n \le q+p\lceil{\tfrac{q}{p}\rceil}-1$ while Case 3.\ additionally implies
  $$H^d(q+p\lceil{\tfrac{q}{p}\rceil}-1,p,q)=H^s(q+p\lceil{\tfrac{q}{p}\rceil}-1,p,q)=2\lceil{\tfrac{q}{p}\rceil}-1.$$
  Consequently, $L(h,p,q)=L^s(h,p,q)$ if $h < 2\lceil{\frac{q}{p}}\rceil-1$, i.e., if $ \frac{q}{p} > \ceil{\frac{h}{2}}$.
  Moreover, if $\frac32 < \frac{q}{p}$, then $3q>q+3p-1$,
  and therefore $H(n,p,q)=H^d(n,p,q)$ for $n \ge q+p\lceil{\tfrac{q}{p}\rceil}-1$ due to Cases 0.\ and 5.--7.
  Hence, if $h \ge 2\lceil{\tfrac{q}{p}}\rceil-1$, i.e., $\frac32 < \frac{q}{p}<\ceil{\frac{h}{2}}$, then $L(h,p,q)=L^d(h,p,q)$.
  
  Now, it suffices to consider the case of $\frac{q}{p} < \frac32 < \ceil{\frac{h}{2}}$.
  Then, by Cases 5., 7., and 0., $H(n,p,q)=H^d(n,p,q)$ for $n\ge q+3p-1$.
  Case 4.\ additionally yields
   $H^d(q+3p-1,p,q)=H^s(q+3p-1,p,q)=5$, so $L(h,p,q)=L^d(h,p,q)$ if $h\ge 5$.
  Moreover,  by Case 4., $H(n,p,q)=H^d(n,p,q)$ for $q+p-1\le n < 3q$, so $L(3,p,q) = L^d(3,p,q)$ due to $H^s(3q,p,q) \ge 4$.
  Finally, we note that Case 3.\ yields $H(n,p,q)=H^s(n,p,q)$ for $3q \le n \le 3p+q-1$,
  so $L(4,p,q)=L^s(4,p,q)$ due to $H(3q-1,p,q)\le H^s(3q-1,p,q)\le 4$.
 \end{proof}

The remaining cases have already been well understood:
\begin{fact}[\cite{ShurGamzova04,DBLP:journals/jda/Blanchet-SadriMS12}]\label{fct:gcd}
If $p,q>1$ are integers such that $\gcd(p,q)\notin \{p,q\}$,
then
$$L(h,p,q) = \gcd(p,q)\cdot L\left(h,\tfrac{p}{\gcd(p,q)},\tfrac{q}{\gcd(p,q)}\right).$$
\end{fact}
\begin{fact}[\cite{DBLP:conf/mfcs/ShurK01}]\label{fct:two}
If $q,h$ are integers such that $q>2$, $2\nmid q$, and $h\ge 0$,
then
$$L(h,2,q) = (2p+1)\floor{\tfrac{h}{p}}+h \bmod p.$$
\end{fact}

The results above lead to our first algorithm for computing $L(h,p,q)$.

\begin{corollary}\label{cor:comph}
Given integers $p,q>1$ such that $\gcd(p,q)\notin \{p,q\}$ and an integer $h\ge 0$,
the value $L(h,p,q)$ can be computed in $\Oh(h+\log p +\log q)$ time.
\end{corollary}
\begin{proof}
First, we apply \cref{fct:gcd} to reduce the computation to $L(h,p',q')$ such that $\gcd(p',q')=1$
and, without loss of generality, $1<p'<q'$.
This takes $\Oh(\log p + \log q)$ time.
If $p'=2$, we use \cref{fct:two},
while for $p'>2$ we rely on the characterization of \cref{thm:wr},
using \cref{lem:ls,lem:simple_hd} for computing $L^s$ and $L^d$, respectively.
The values $\widetilde{G}(h',p',q')$ form a sorted sequence of multiples of $p'$ and $q'$, but not of $p'q'$.
Hence, it takes $\Oh(h)$ time to generate them for $0\le h'\le h$. The overall running time is $\Oh(h+\log p + \log q)$.
\end{proof}

\section{Faster Algorithm for Evaluating $L$}\label{sec:fast}
A more efficient algorithm for evaluating $L$ relies on the theory of continued fractions; we refer to \cite{Khinchin} and \cite{Farey} for a self-contained yet compact introduction.
A finite continued fraction is a sequence $[\gamma_0;\gamma_1,\ldots,\gamma_m]$,
where $\gamma_0,m\in \mathbb{Z}_{\ge 0}$ and $\gamma_i \in \mathbb{Z}_{\ge 1}$ for $1\le i \le m$.
We associate it with the following rational number:
$$[\gamma_0;\gamma_1,\ldots,\gamma_m]=\gamma_0 + \tfrac{1}{\gamma_1+\tfrac{1}{\ddots+\tfrac{1}{\gamma_m}}}.$$
Depending on the parity of $m$, we distinguish odd and even continued fractions.
Often, an improper continued fraction $[;]=\frac{1}{0}$ is also introduced and assumed to be odd.
Each positive rational number has exactly two representations as a continued fraction, one as an even continued fraction,
and one as an odd continued fraction.
For example, $\frac{5}{7}=[0;1,2,2]=[0;1,2,1,1]$.

Consider a continued fraction $[\gamma_0;\gamma_1,\ldots,\gamma_m]$.
Its \emph{convergents} are continued fractions of the form $[\gamma_0;\gamma_1,\ldots,\gamma_{m'}]$ for $0\le m' < m$,
and $[;]=\frac{1}{0}$.
The \emph{semiconvergents} also include continued fractions of the form $[\gamma_0;\gamma_1,\ldots,\gamma_{m'-1},\gamma'_{m'}]$, 
where $0\le m'\le m$ and $0<\gamma'_{m'}< \gamma_{m'}$.
The two continued fractions representing a positive rational number have the same semiconvergents.
 
\begin{example}
The semiconvergents of $[0;1,2,2]=\frac{5}{7}=[0;1,2,1,1]$ are $[;]=\frac{1}{0}$, $[0;]=\frac{0}{1}$, $[0;1]=\frac{1}{1}$,
 $[0;1,1]=\frac{1}{2}$, $[0;1,2]=\frac{2}{3}$, and
 $[0;1,2,1]=\frac{3}{4}$.
\end{example}

Semiconvergents of $\frac{p}{q}$ can be generated using the (slow) \emph{continued fraction algorithm},
which produces a sequence of \emph{Farey pairs} $(\frac{a}{b},\frac{c}{d})$ such that $\frac{a}{b}<\frac{p}{q}<\frac{c}{d}$.

\begin{center}
\begin{minipage}{0.59\textwidth}\begin{algorithm}[H]
$(\frac{a}{b},\frac{c}{d}):= (\frac{0}{1},\frac{1}{0})$\;
\While{\KwSty{true}}{
Report a Farey pair $(\frac{a}{b},\frac{c}{d})$\;
\lIf{$\frac{a+c}{b+d}<\frac{p}{q}$}{$\frac{a}{b}:= \frac{a+c}{b+d}$}
\lElseIf{$\frac{a+c}{b+d}=\frac{p}{q}$}{\KwSty{break}}
\lElse{$\frac{c}{d}:= \frac{a+c}{b+d}$}
}
\caption{Farey process for a rational number $\frac{p}{q}>0$}
\end{algorithm}
\end{minipage}
\end{center}
\begin{example}
For $\frac{p}{q}=\frac{5}{7}$, the Farey pairs are
$(\frac{0}{1},\frac{1}{0})\leadsto (\frac{0}{1},\frac{1}{1})\leadsto (\frac{1}{2},\frac{1}{1})\leadsto (\frac{2}{3},\frac{1}{1})\leadsto (\frac{2}{3},\frac{3}{4})$.
The process terminates at $\frac{2+3}{3+4}=\frac{5}{7}$.
\end{example}

Consider the set $\Fr = \{\frac{a}{b} : a,b\in \mathbb{Z}_{\ge 0},\, \gcd(a,b)=1\}$ of reduced fractions (including $\frac{1}{0}$).
We denote $\Fr_k = \{\frac{a}{b}\in \Fr :  a+b \le k\}$ and, for each $x\in \mathbb{R}_+$:
$$\pred_k(x)=\max\{a \in \Fr_k\,:\,a \le x\}\quad\text{and}\quad\succ_k(x)=\min\{a \in \Fr_k\,:\,a \ge x\}.$$
We say that $\frac{a}{b}<x$ is a \emph{best left approximation} of $x$ if $\frac{a}{b}=\pred_k(x)$
for some $k\in \mathbb{Z}_{\ge 0}$. Similarly, $\frac{c}{d}>x$ is
a \emph{best right approximation} of $x$ if $\frac{c}{d}=\succ_k(x)$.
\begin{example}
We have 
$\Fr_7= (\tfrac{0}{1},\tfrac{1}{6},\tfrac{1}{5},\tfrac{1}{4},\tfrac{1}{3},\tfrac{2}{5},\tfrac{1}{2},\tfrac{2}{3},\tfrac{3}{4},\tfrac{1}{1},\frac{4}{3},\frac{3}{2},\tfrac{2}{1},\tfrac{5}{2},\tfrac{3}{1},\tfrac{4}{1},\tfrac{5}{1},\tfrac{6}{1},\tfrac{1}{0}).$
Here, $\pred_7(\tfrac{5}{7})\,=\, \frac{2}{3}$ and $\succ_7(\tfrac{5}{7}) = \tfrac{3}{4}$
are best approximations of $\frac{5}{7}$.
\end{example}
We heavily rely  on the following extensive characterization of semiconvergents:
\begin{fact}[{\cite{Khinchin}}, {\cite[Theorem 3.3]{Ravenstein1988}}, {\cite[Theorem 2]{Farey}}]\label{fct:farey}
Let $\frac{p}{q}\in \Fr\setminus\{\frac{1}{0},\frac{0}{1}\}$. The following conditions are equivalent
for reduced fractions $\frac{a}{b}<\frac{p}{q}$:
\begin{enumerate}[label={\rm(\alph*)}]
  \item the Farey process for $\frac{p}{q}$ generates a pair $(\frac{a}{b},\frac{c}{d})$ for some $\frac{c}{d}\in \Fr$,
  \item $\frac{a}{b}$ is an even semiconvergent of $\frac{p}{q}$,
  \item $\frac{a}{b}$ is a best left approximation of $\frac{p}{q}$,
  \item $b=\floor{\frac{aq}{p}}+1$ and $aq \bmod p > iq \bmod p$ for $0\le i < a$.
\end{enumerate}
By symmetry, the following conditions are equivalent
for reduced fractions $\frac{c}{d}>\frac{p}{q}$:
\begin{enumerate}[label={\rm(\alph*)}]
  \item the Farey process for $\frac{p}{q}$ generates a pair $(\frac{a}{b},\frac{c}{d})$ for some $\frac{a}{b}\in \Fr$,
  \item $\frac{c}{d}$ is an odd semiconvergent of $\frac{p}{q}$,
  \item $\frac{c}{d}$ is a best right approximation of $\frac{p}{q}$,
  \item $c=\floor{\frac{dp}{q}}+1$ and $dp \bmod q > ip \bmod q$ for $0\le i<d$.
\end{enumerate}
\end{fact}

\begin{example}
For $\frac{p}{q}=\frac{5}{7}$, the prefix maxima of $(iq \bmod p)_{i=0}^{p-1}=(0,2,4,1,3)$ 
are attained for $i=0,1,2$ (numerators of $\frac{0}{1},\frac{1}{2},\frac{2}{3}$)
while the prefix maxima of $(ip \bmod q)_{i=0}^{q-1}=(0,5,3,1,6,4,2)$
are attained for $i=0,1,4$ (denominators $\frac{1}{0},\frac{1}{1},\frac{3}{4}$).
\end{example}

Due to \cref{fct:farey}, the best approximations can be efficiently computed using the \emph{fast} continued fraction algorithm; see \cite{Farey}.
\begin{corollary}\label{cor:fast_farey}
Given $\frac{p}{q}\in \Fr$ and a positive integer $k$, $1\le k< p+q$, the values $\pred_k(\frac{p}{q})$ and $\succ_k(\frac{p}{q})$ can be computed in $\Oh(\log k)$ time.
\end{corollary}

Next, we characterize the function $L^d$.
\begin{lemma}\label{lem:ld}
Let $p,q>2$ be relatively prime integers and let $h < p+q-3$.
If $\frac{a}{b} = \pred_{h+3}(\frac{p}{q})$
and $\frac{c}{d} = \succ_{h+3}(\frac{p}{q})$,
then, assuming $G(-1,p,q)=0$:
$$L^d(h,p,q)=\begin{cases}
\widetilde{G}(a+b-2,p,q)+\widetilde{G}(c+d-2,p,q) & \text{if }a+b+c+d = h+4,\\
\widetilde{G}(h+2,p,q) & \text{otherwise.}
\end{cases}$$
\end{lemma}
\begin{proof}
Let us start with a special case of $\frac{a}{b}=\frac{0}{1}$.
Then $\frac{c}{d}=\frac{1}{h+2}$, so $q > (h+2)p$ and $\widetilde{G}(k,p,q)=(k+1)p$ for $k \le h+1$.
Consequently, by \cref{lem:simple_hd},
$$L^d(h,p,q)=\min_{k=0}^{h}\left(\widetilde{G}(k,p,q)+\widetilde{G}(h-k,p,q)\right)=(h+2)p.$$
Due to $a+b+c+d=0+1+1+h+2 = h+4$, this is equal to the claimed value of $\widetilde{G}(-1,p,q)+\widetilde{G}(h+1,p,q)=0+(h+2)p.$
Symmetrically, the lemma holds if $\frac{c}{d}=\frac{1}{0}$.
Thus, below we assume $\frac{1}{h+2}< \frac{p}{q}<\frac{h+2}{1}$.

By \cref{fct:farey}, $\frac{a+c}{b+d}$ is a best (left or right) approximation of $\frac{p}{q}$,
so $\max(a+b,c+d)\le h+3 < a+b+c+d$.
Moreover, 
$$G(aq,p,q)=\floor{\tfrac{aq}{p}}+\floor{\tfrac{aq}{q}}=b-1+a\quad\text{and}
\quad G(dp,p,q)=\floor{\tfrac{dp}{p}}+\floor{\tfrac{dp}{q}}=d+c-1,$$
so $\widetilde{G}(a+b-2,p,q)+\widetilde{G}(c+d-2,p,q)=aq+dp$

First, suppose that $a+b+c+d< h+4$.
Assume without loss of generality that $\widetilde{G}(h+2,p,q)=\alpha p$ is a multiple of $p$.
Note that $d<\alpha<b+d$ due to
 $$G((b+d)p,p,q)=b+d+\floor{\tfrac{(b+d)p}{q}}\ge a+b+c+d-1\ge h+4.$$
Consequently, \cref{fct:farey} yields $\alpha p \bmod q < dp \bmod q$. Hence
 $$G((\alpha-d)p,p,q)=(\alpha-d)+\floor{\tfrac{\alpha p - dp}{q}}=(\alpha-d)+\floor{\tfrac{\alpha p}{q}}+\floor{\tfrac{-dp}{q}}=h+3-c-d,$$
and therefore
$$L^d(h,p,q)\ge \widetilde{G}(h+2-c-d,p,q)+\widetilde{G}(c+d-2,p,q)=(\alpha-d)p+dp= \widetilde{G}(h+2,p,q).$$
On the other hand, \cref{lem:hdb} yields
$H^d(\alpha p,p,q)\ge G(\alpha p, p, q)-2 = h+1$, so $L^d(h,p,q)\le \widetilde{G}(h+2,p,q)$.

Finally, suppose that $a+b+c+d=h+4$.
\cref{lem:simple_hd} immediately yields $L^d(h,p,q)\ge \widetilde{G}(a+b-2,p,q)+\widetilde{G}(c+d-2,p,q)=aq+dp$.
For the proof of the inverse inequality, let us take $k$ such that $L^d(h,p,q)=\widetilde{G}(k,p,q)+\widetilde{G}(h-k,p,q)$,
and define $x=\widetilde{G}(k,p,q)$ and $y=\widetilde{G}(h-k,p,q)$.
Consequently,
\begin{multline*}
a+b+c+d= h+4= \floor{\tfrac{x-1}{p}}+\floor{\tfrac{x-1}{q}}+\floor{\tfrac{y-1}{p}}+\floor{\tfrac{y-1}{q}}+4=
\ceil{\tfrac{x}{p}}+\ceil{\tfrac{y}{p}}+\ceil{\tfrac{x}{q}}+\ceil{\tfrac{y}{q}}\ge\\
\ceil{\tfrac{x+y}{p}}+\ceil{\tfrac{x+y}{q}}\ge \ceil{\tfrac{aq+dp}{p}}+\ceil{\tfrac{aq+dp}{q}}=
d+\ceil{\tfrac{aq}{p}}+a+\ceil{\tfrac{dp}{q}}=d+b+a+c.$$
\end{multline*}
Each intermediate inequality must therefore be an equality, so we conclude that
$$\ceil{\tfrac{x}{p}}+\ceil{\tfrac{y}{p}}=\ceil{\tfrac{x+y}{p}}=\ceil{\tfrac{aq+dp}{p}}=b+d\quad\text{and}\quad
\ceil{\tfrac{x}{q}}+\ceil{\tfrac{y}{q}}=\ceil{\tfrac{x+y}{q}}=\ceil{\tfrac{aq+dp}{q}}=a+c.$$

If $p\mid x$ and $p \mid y$, then $\frac{x+y}{p} = b+d$,
so $\ceil{\frac{(b+d)p}{q}}=a+c$.
Hence $\frac{a+c}{b+d}\ge \frac{p}{q}$, and consequently
$\frac{a+c}{b+d}$ is either a right semiconvergent of $\frac{p}{q}$ or is equal to $\frac{p}{q}$.
In both cases, \cref{fct:farey} implies $(-(b+d)p)\bmod q < \min((-x)\bmod q, (-y)\bmod q)$.
This lets us derive a contradiction:
$$\ceil{\tfrac{x}{q}}+\ceil{\tfrac{y}{q}}=\tfrac{x+y+(-x)\bmod q+(-y)\bmod q}{q}>\tfrac{(b+d)p+2((-(b+d)p)\bmod q)}{q}\ge\ceil{\tfrac{(b+d)p}{q}}.$$
Symmetrically, $q\mid x$ and $q \mid y$ yields an analogous contradiction.

Thus, without loss of generality we may assume $p \mid x$ and $q\mid y$.
However, the conditions  $x+y \ge aq+dp$ and $\ceil{\frac{x+y}{p}}=\ceil{\frac{aq+dp}{p}}$
yield $(-y)\bmod p = (-(x+y))\bmod p \le (-(aq+dp))\bmod p = (-dp)\bmod p$.
By \cref{fct:farey}, this implies $y=dp$. Symmetrically, $x=aq$.
Thus, $L^d(h,p,q)=aq+dp$, as claimed.
\end{proof}

\cref{lem:ld} applies to $h<p+q-3$;
the following fact lets us deal with $h\ge p+q-3$.
It appeared in~\cite{DBLP:journals/jda/Blanchet-SadriMS12}, but we provide an alternative proof for completeness.
\begin{fact}[{\cite[Theorem 4]{DBLP:journals/jda/Blanchet-SadriMS12}}]\label{fct:reduce}
Let $p,q$ be relatively prime positive integers.
For each $h\ge 0$, we have 
$$L^d(h,p,q) = L^d(h \bmod (p+q-2),p,q) + \floor{\tfrac{h}{p+q-2}} \cdot pq.$$
Moreover, $L^d(p+q-3,p,q)=pq$.
\end{fact}
 \begin{proof}
First, note that $\widetilde{G}(k,p,q)+\widetilde{G}(p+q-3-k,p,q)=pq$ holds for $0\le k \le p+q-3$.
Hence, $L^d(p+q-3,p,q)=pq$  holds as claimed due to \cref{lem:simple_hd}.

For the first part of the statement, it suffices to prove that $H^d(n+pq,p,q)=H^d(n,p,q)+p+q-2$ for each $n\ge q$.
The function $G$ satisfies an analogous equality, so \cref{lem:simple_hd} immediately yields
$H^d(n+pq,p,q)\le p+q+2+H^d(n,p,q)$. The other inequality also follows from \cref{lem:simple_hd}
unless each optimum value $l$ for $n+pq$ satisfies $n \le l < pq$.
However, for such $l$ (and $q<n<pq$), we have
\begin{multline*}
G(l,p,q)+G(n+pq-l-1,p,q)=\floor{\tfrac{l}{p}}+\floor{\tfrac{l}{q}}+\floor{\tfrac{n+pq-l-1}{p}}+\floor{\tfrac{n+pq-l-1}{p}}\ge\\
\floor{\tfrac{n+pq}{p}}-1 + \floor{\tfrac{n+pq}{q}}-1 = G(n+pq,p,q)+G(0,p,q),
\end{multline*}
a contradiction. This concludes the proof.
\end{proof}

 \begin{theorem}\label{thm:complogh}
Given integers $p,q\ge 1$ such that $\gcd(p,q)\notin \{p,q\}$ and an integer $h\ge 0$,
the value $L(h,p,q)$ can be computed in $\Oh(\log p +\log q)$ time.
  \end{theorem}
  \begin{proof}
  We proceed as in the proof of \cref{cor:comph},
  except that we apply \cref{fct:reduce,lem:ld} to compute $L^d(h,p,q)$.
  \cref{fct:reduce} reduces the problem to determining $L^d(h',p,q)$,
  where $h'=h \bmod (p+q-2)$.
  We use \cref{cor:fast_farey} to compute $\pred_{h'+3}(\frac{p}{q})$
  and $\succ_{h'+3}(\frac{p}{q})$ in $\Oh(\log h')$ time.
  The values $\widetilde{G}(r,p,q)$ can be determined in $\Oh(\log r)$ time using binary search (restricted to multiples of $p$ or $q$).
  The overall running time for $L^d(h,p,q)$ is $\Oh(\log h')=\Oh(\log p+ \log q)$,
  so for $L(h,p,q)$ it is also $\Oh(\log p+ \log q)$.
  \end{proof}

\section{Closed-Form Formula for $L(h,\cdot,\cdot)$}\label{sec:closed}

In this section we show how to compute a compact representation of the function
$L(h,\cdot,\cdot)$ in $\Oh(h\log h)$ time. 
We start with such representations for $\widetilde{G}$ and $L^d$.

Assume that $h < p+q - 3$.
  For $0<i\le h+4$, let us define fractions
 $$l_i = \tfrac{i-1}{h+4-i},\ \  
  m_i=\tfrac{i}{h+4-i},$$ 
called the $h$-special points and the $h$-middle points, respectively.
Now, The function $\widetilde{G}$ can be expressed as follows (see~\cref{fig:g}):

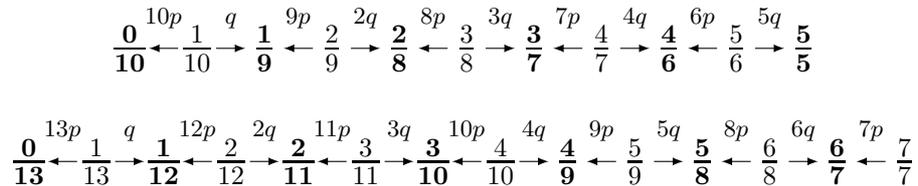
\begin{figure}[b!]
\begin{center}
\Large
\begin{tikzpicture}
[xscale=0.886,
left/.style={latex-,shorten <=0.25cm,shorten >=0.25cm},
right/.style={-latex,shorten <=0.25cm,shorten >=0.25cm},
label/.style={above=0.1cm, font=\small}]
  \draw (-3,0) node {$\bf \frac{0}{10}$};
  \draw (-2,0) node {$\frac{1}{10}$};
  \draw (-1,0) node {$\bf \frac{1}{9}$};
  \draw (0,0) node {$\frac{2}{9}$};
  \draw (1,0) node {$\bf \frac{2}{8}$};
  \draw (2,0) node {$\frac{3}{8}$};
  \draw (3,0) node {$\bf \frac{3}{7}$};
  \draw (4,0) node {$\frac{4}{7}$};
  \draw (5,0) node {$\bf \frac{4}{6}$};
  \draw (6,0) node {$\frac{5}{6}$};
  \draw (7,0) node {$\bf \frac{5}{5}$};
  \draw[left] (-3,0) -- node[label] {$10p$} (-2,0);
  \draw[right] (-2,0) -- node[label] {$q$} (-1,0);
  \draw[left] (-1,0) -- node[label] {$9p$} (0,0);
  \draw[right] (0,0) -- node[label] {$2q$} (1,0);
  \draw[left] (1,0) -- node[label] {$8p$} (2,0);
  \draw[right] (2,0) -- node[label] {$3q$} (3,0);
  \draw[left] (3,0) -- node[label] {$7p$} (4,0);
  \draw[right] (4,0) -- node[label] {$4q$} (5,0);
  \draw[left] (5,0) -- node[label] {$6p$} (6,0);
  \draw[right] (6,0) -- node[label] {$5q$} (7,0);
  
  \begin{scope}[yshift=-1.5cm,xshift=-1.5cm]
    \draw (-3,0) node {$\bf \frac{0}{13}$};
  \draw (-2,0) node {$\frac{1}{13}$};
  \draw (-1,0) node {$\bf \frac{1}{12}$};
  \draw (0,0) node {$\frac{2}{12}$};
  \draw (1,0) node {$\bf \frac{2}{11}$};
  \draw (2,0) node {$\frac{3}{11}$};
  \draw (3,0) node {$\bf \frac{3}{10}$};
  \draw (4,0) node {$\frac{4}{10}$};
  \draw (5,0) node {$\bf \frac{4}{9}$};
  \draw (6,0) node {$\frac{5}{9}$};
  \draw (7,0) node {$\bf \frac{5}{8}$};
  \draw (8,0) node {$\frac{6}{8}$};
  \draw (9,0) node {$\bf \frac{6}{7}$};
  \draw (10,0) node {$\frac{7}{7}$};
  \draw[left] (-3,0) -- node[label] {$13p$} (-2,0);
  \draw[right] (-2,0) -- node[label] {$q$} (-1,0);
  \draw[left] (-1,0) -- node[label] {$12p$} (0,0);
  \draw[right] (0,0) -- node[label] {$2q$} (1,0);
  \draw[left] (1,0) -- node[label] {$11p$} (2,0);
  \draw[right] (2,0) -- node[label] {$3q$} (3,0);
  \draw[left] (3,0) -- node[label] {$10p$} (4,0);
  \draw[right] (4,0) -- node[label] {$4q$} (5,0);
  \draw[left] (5,0) -- node[label] {$9p$} (6,0);
  \draw[right] (6,0) -- node[label] {$5q$} (7,0);
   \draw[left] (7,0) -- node[label] {$8p$} (8,0);
  \draw[right] (8,0) -- node[label] {$6q$} (9,0);
   \draw[left] (9,0) -- node[label] {$7p$} (10,0);
  
  \end{scope}
\end{tikzpicture}
\end{center}
\caption{
Graphical representations of the closed-form formulae for $\widetilde{G}(9,p,q)$ (above)
and $\widetilde{G}(12,p,q)$ (below) for $p<q$: partitions of $[0,1]$ into intervals w.r.t. $p/q$ and
linear functions of $p$ and $q$ for each interval.
The respective special points are shown in bold.
}\label{fig:g}
\end{figure}%

\begin{lemma}\label{lem:g}
If $ \gcd(p,q)=1$ and $h< p+q-3$, then
$$\widetilde{G}(h+2,p,q) =\begin{cases}
\; (h+4-i) \cdot p &\text{ if }\ l_i \le \frac{p}{q}\le m_i, \\
\; i \cdot q &\text{ if }\ m_i \le \frac{p}{q}\le l_{i+1}.
\end{cases}
$$
\end{lemma}
\begin{proof}
Note that $\widetilde{G}(h+2,p,q)=n$ is equivalent to $G(n-1,p,q)\le h+2 < G(n,p,q)$.
Additionally, observe that $\widetilde{G}(h+2,p,q)$ is a multiple of $p$ or $q$.
We have two cases.
\paragraph{Case 1:} The condition $\widetilde{G}(h+2,p,q)=j\cdot q$ for $j\in \mathbb{Z}_{>0}$
is equivalent to:
$$\floor{\tfrac{jq}{p}}+j \ge h+3 \quad\quad\text{and}\quad\quad \floor{\tfrac{jq-1}{p}}+j-1\le h+2,$$
i.e.,
$$\floor{\tfrac{jq}{p}}\ge h+3-j \quad\quad \text{and}\quad\quad \ceil{\tfrac{jq}{p}}= \floor{\tfrac{jq-1}{p}}+1 \le h+4-j.$$
In other words, we have
$h+3-j \le \frac{jq}{p} \le h+4-j$, i.e.,
$$m_j=\tfrac{j}{h+4-j}\le \tfrac{p}{q} \le \tfrac{j}{h+3-j}=l_{j+1}.$$

\paragraph{Case 2:} The condition $\widetilde{G}(h+2,p,q)=j\cdot p$ for $j\in \mathbb{Z}_{>0}$
is equivalent to:
$$\floor{\tfrac{j p}{q}}+j \ge h+3 \quad\quad\text{and}\quad\quad \floor{\tfrac{jp-1}{q}}+j-1\le h+2,$$
i.e.,
$$\floor{\tfrac{jp}{q}}\ge h+3-j \quad\quad \text{and}\quad\quad \ceil{\tfrac{jq}{p}}= \floor{\tfrac{jp-1}{q}}+1 \le h+4-j.$$
In other words, we have
$h+3-j \le \frac{jp}{q} \le h+4-j$, i.e.,
$$l_{h+4-j}=\tfrac{h+3-j}{j}\le \tfrac{p}{q} \le \tfrac{h+4-j}{j}=m_{h+4-j}.$$

\smallskip
The family of intervals $[m_i,l_{i+1}]$ and $[l_i,m_i]$ has
the property that any two distinct intervals in this family have disjoint interiors.
Hence, the values of $\widetilde{G}(h,p,q)$ are as claimed.     \end{proof}

Combined with \cref{lem:ld}, \cref{lem:g} yields a closed-form formula for $L^d$.
Note that for each $i$, we have $l_i \le \pred_{h+3}(m_i)\le m_i \le \succ_{h+3}(m_i)\le l_{i+1}$,
but none of the inequalities is strict in general.
In particular, $\pred_{h+3}(m_i)= m_i = \succ_{h+3}(m_i)$
if $\gcd(i, h+4-i)>1$.%
\begin{corollary}\label{cor:main1}
Let $p,q$ be relatively prime positive integers and let $h\le p+q-3$ be a non-negative integer.
Suppose that $l_i \le \frac{p}{q} \le l_{i+1}$ and define reduced fractions
 $\frac{a_i}{b_i} = \pred_{h+3}(m_i)$ and $\frac{c_i}{d_i}=\succ_{h+3}(m_i)$.
 Then:
$$L^d(h,p,q) =
\begin{cases}
(h+4-i) \cdot p & \text{ if } l_i \le \frac{p}{q} \le \frac{a_i}{b_i},\\
a_iq + d_i p & \text{ if }\frac{a_i}{b_i} < \frac{p}{q} < \frac{c_i}{d_i},\\
i \cdot q & \text{ if } \frac{c_i}{d_i} \le \frac{p}{q} \le l_{i+1}.
\end{cases}
$$
\end{corollary}
 \begin{proof}
First, observe that for $h=p+q-3$, we have $\frac{p}{q}=l_{p+1}$
and $m_p < \frac{p}{q} < m_{p+1}$, so $\frac{c_p}{d_p}\le \frac{p}{q}\le \frac{a_{p+1}}{b_{p+1}}$.
As claimed, $L^d(h,p,q)=(h+4-(p+1))\cdot p= p \cdot q$.

Below, we assume $h < p+q-3$.
Let $\frac{a}{b}=\pred_{h+3}(\frac{p}{q})$ and $\frac{c}{d}=\succ_{h+3}(\frac{p}{q})$.
We shall prove that $a+b+c+d=h+4$ if and only if $\frac{a_i}{b_i} < \frac{p}{q} < \frac{c_i}{d_i}$
for some $i$.

First, suppose that $a+b+c+d=h+4$. This means that $\frac{a+c}{b+d}\in \Fr_{h+4}\setminus \Fr_{h+3}$,
so $\frac{a+c}{b+d}=m_i$ for some $i$, and therefore $\frac{a}{b}=\frac{a_i}{b_i}$ and $\frac{c}{d}=\frac{c_i}{d_i}$.
Consequently, $\frac{a_i}{b_i} < \frac{p}{q} < \frac{c_i}{d_i}$.
In the other direction, $\frac{a_i}{b_i} < \frac{p}{q} < \frac{c_i}{d_i}$
implies $\frac{a}{b}=\frac{a_i}{b_i}$ and $\frac{c}{d}=\frac{c_i}{d_i}$,
so $\frac{a}{b} < m_i < \frac{c}{d}$. By \cref{fct:farey}, this yields $a+b+c+d \le h+4$.
Moreover, $\frac{a+c}{b+d}\notin \Fr_{h+3}$, so $a+b+c+d=h+4$.

Since $G(a_iq,p,q) = a_i + b_i - 1$ and $G(d_i p,p,q ) = c_i + d_i - 1$ by \cref{fct:farey},
we have $a_iq+d_ip=\widetilde{G}(a_i+b_i-2,p,q)+\widetilde{G}(c_i+d_i-2,p,q)$.
Now, \cref{lem:g,lem:ld} yield the final formula.
\end{proof}

\begin{figure}[t]
\begin{center}
\vspace{-.35cm}
\Large
\begin{tikzpicture}[xscale=0.951,
left/.style={latex-,shorten <=0.25cm,shorten >=0.25cm},
right/.style={-latex,shorten <=0.25cm,shorten >=0.25cm},
middle/.style={latex-latex,shorten <=0.25cm,shorten >=0.25cm},
init/.style={shorten <=0.25cm,shorten >=0.25cm},
label/.style={above=0.25cm, font=\small}]

\begin{scope}[xshift=1.6cm]
  \draw (-.8,0) node {$0$};
  \draw (0.4,0) node {$\frac{1}{4}\!=$};
  \draw (1,0) node {$\bf \frac{2}{8}$};
  \draw (2,0) node {$\frac{1}{3}$};
  \draw (3,0) node {$\frac{2}{5}$};
  \draw (4,0) node {$\bf \frac{3}{7}$};
  \draw (5,0) node {$\frac{1}{2}$};
  \draw (6,0) node {$\frac{3}{5}$};
  \draw (7,0) node {$\bf \frac{4}{6}$};
  \draw (8,0) node {$\frac{4}{5}$};
  \draw (9,0) node {$\frac{1}{1}$};
  \draw (9.6,0) node {$=\!\bf\frac{5}{5}$};
  \draw[init] (-.8,0) -- node[label] {$q\!+\!4p$} (0.2,0);
  \draw[left] (1,0) -- node[label] {$8p$} (2,0);
  \draw[middle] (2,0) -- node[label] {$q\!+\!5p$} (3,0);
  \draw[right] (3,0) -- node[label] {$3q$} (4,0);
  \draw[left] (4,0) -- node[label] {$7p$} (5,0);
  \draw[middle] (5,0) -- node[label] {$q\!+\!5p$} (6,0);
  \draw[right] (6,0) -- node[label] {$4q$} (7,0);
  \draw[left] (7,0) -- node[label] {$6p$} (8,0);
  \draw[middle] (8,0) -- node[label] {$4q\!+\!p$} (9,0);
  \end{scope}
  \begin{scope}[yshift=-1.5cm]
    \draw (0,0) node {$0$};
  \draw (1,0) node {$\frac{1}{5}$};
  \draw (2,0) node {$\frac{1}{4}$};
  \draw (3,0) node {$\frac{2}{7}$};
  \draw (4,0) node {$\bf \frac{3}{10}$};
  \draw (5,0) node {$\frac{2}{5}$};
  \draw (6,0) node {$\bf \frac{4}{9}$};
  \draw (7,0) node {$\frac{1}{2}$};
  \draw (8,0) node {$\frac{4}{7}$};
  \draw (9,0) node {$\bf\frac{5}{8}$};
  \draw (10,0) node {$\frac{3}{4}$};
  \draw (11,0) node {$\bf \frac{6}{7}$};
  \draw (12,0) node {$\frac{7}{7}$};
  \draw[init] (0,0) -- node[label] {$q\!+\!6p\!-\!\!1$} (1,0);
  \draw[left] (1,0) -- node[label] {$11p$} (2,0);
  \draw[middle] (2,0) -- node[label] {$q\!+\!7p$} (3,0);
  \draw[right] (3,0) -- node[label] {$3q$} (4,0);
  \draw[left] (4,0) -- node[label] {$10p$} (5,0);
  \draw[right] (5,0) -- node[label] {$4q$} (6,0);
  \draw[left] (6,0) -- node[label] {$9p$} (7,0);
  \draw[middle] (7,0) -- node[label] {$q\!+\!7p$} (8,0);
  \draw[right] (8,0) -- node[label] {$5q$} (9,0);
  \draw[left] (9,0) -- node[label] {$8p$} (10,0);
  \draw[right] (10,0) -- node[label] {$6q$} (11,0);
  \draw[left] (11,0) -- node[label] {$7p$} (12,0);
  
  \end{scope}

  \begin{scope}[yshift=1.8cm, xshift=2cm]
    \draw (0,0) node {$\frac{\bf a}{\bf b}$};
    \draw (2,0) node {$\frac{c}{d}$};
    \draw[left] (0,0) -- node[label,above=0cm] {$b \cdot p$} node[below=0.2cm] {\small left subinterval} (2,0);

    \begin{scope}[xshift=3cm]
    \draw (0,0) node {$\frac{a}{b}$};
    \draw (2,0) node {$\frac{c}{d}$};
    \draw[middle] (0,0) -- node[label,above=0cm] {$a\cdot q + d \cdot p$} node[below=0.2cm] {\small middle subinterval} (2,0);
    \end{scope}

    \begin{scope}[xshift=6cm]
    \draw (0,0) node {$\frac{a}{b}$};
    \draw (2,0) node {$\frac{\bf c}{\bf d}$};
    \draw[right] (0,0) -- node[label,above=0cm] {$c\cdot q$} node[below=0.2cm] {\small right subinterval} (2,0);
    \end{scope}

  \end{scope}
\end{tikzpicture}
\vspace{-.35cm}
\end{center}

\caption{
Graphical representations of the closed-form formulae for $L(7,p,q)$ (middle)
and $L(10,p,q)$ (below). Compared to $\widetilde{G}(9,p,q)$ and $\widetilde{G}(12,p,q)$, respectively,
an initial subinterval and several middle subintervals are added. 
A general pattern for the left, middle, and right subintervals, 
is presented above. However, the left subinterval $(\tfrac1{5},\tfrac1{4})$ within $L(10,p,q)$ is an exception
because is has been trimmed by the initial interval.}\label{fig:l}
\end{figure}

\begin{theorem}\label{thm:table}
Let $2<p<q$ be relatively prime and let $4<h< p+q-2$.
Suppose that $l_i \le \frac{p}{q} \le l_{i+1}$ and define reduced fractions
 $\frac{a_i}{b_i} = \pred_{h+3}(m_i)$ and $\frac{c_i}{d_i}=\succ_{h+3}(m_i)$.
 Then:
$$L(h,p,q)=
\begin{cases}
\ceil{\tfrac{h+1}{2}} p + q - (h+1)\bmod 2 & \text{if}\ 0<\frac{p}{q}< 1/\ceil{\tfrac{h}{2}} \ \text{else}\\
(h+4-i) \cdot p & \text{if } l_i \le \frac{p}{q} \le \frac{a_i}{b_i},\\
a_iq + d_i p & \text{if }\frac{a_i}{b_i} < \frac{p}{q} < \frac{c_i}{d_i},\\
i \cdot q & \text{if } \frac{c_i}{d_i} \le \frac{p}{q} \le l_{i+1}.
\end{cases}
$$
This compact representation of $L(h,p,q)$ (see \cref{fig:l} for an example)
for a given $h$ has size $\Oh(h)$ and can be computed in time $\Oh(h \log h)$.
\end{theorem}
\begin{proof}
The formula follows from the formulae for $L^s$ (\cref{lem:ls}) and $L^d$ (\cref{cor:main1})
combined using \cref{thm:wr}.
To compute the table for $L$ efficiently, we determine $\frac{a_i}{b_i}=\pred_{h+3}(m_i)$ and $\frac{c_i}{d_i}=\succ_{h+3}(m_i)$ using 
\cref{cor:fast_farey}.
\end{proof}

\section{Relation to Standard Sturmian Words}
For a finite \emph{directive sequence} $\gamma=(\gamma_1,\ldots,\gamma_m)$ of positive integers,
a Sturmian word $\St(\gamma)$ is recursively defined as $X_m$, where
$X_{-1}=\q$, $X_0 = \p$, and $X_i = X_{i-1}^{\gamma_i}X_{i-2}$ for $1\le i \le m$; see~\cite[Chapter~2]{Lothaire2002}.
We classify directive sequences $\gamma$ (and the Sturmian words $\St(\gamma)$) into \emph{even} and \emph{odd} 
based on the \emph{parity} of $m$.
\begin{observation}
Odd Sturmian words of length at least 2 end with $\p\q$, while even Sturmian words of length at least 2 end with $\q\p$.
\end{observation}
For a directive sequence $\gamma=(\gamma_1,\ldots,\gamma_m)$,
we define $\fr(\gamma)=[0;\gamma_1,\ldots,\gamma_m]$.
\begin{fact}[{\cite[Proposition 2.2.24]{Lothaire2002}}]
If $\fr(\gamma)=\frac{p}{q}$,
then $\St(\gamma)$ contains $p$ characters $\q$ and $q$ characters $\p$.
\end{fact}
\begin{example}
We have $\frac{5}{7}=[0;1,2,2]=[0;1,2,1,1]$, so the Sturmian words with 5 $\q$'s and 7 $\p$'s are:
$\St(1,2,2)=\p\q\p\q\p\p\q\p\q\p\p\q$ and $\St(1,2,1,1)=\p\q\p\q\p\p\q\p\q\p\q\p$.
\end{example}
For relatively prime integers $1< p<q$, we define $\St_{p,q}$ as a Sturmian word with $\fr(\gamma)=\frac{p}{q}$.
Note that we always have two possibilities for $\St_{p,q}$ (one odd and one even), but they differ in the last two positions only.
In fact, the first $p+q-2$ characters of $\St_{p,q}$ are closely related to the values $\widetilde{G}(i,p,q)$.
\begin{fact}[{\cite[Proposition 2.2.15]{Lothaire2002}}]
Let $1< p<q$ be relatively prime integers.
If $i\le p+q-3$, then
$$\St_{p,q}[i] = \begin{cases}
\p & \text{if }p \mid \widetilde{G}(i,p,q),\\
\q & \text{if }q \mid \widetilde{G}(i,p,q).
\end{cases}
$$\end{fact}
As a result, the values $\widetilde{G}(i,p,q)$ can be derived from $\St_{p,q}$; see~\cref{tab:sturmian}.

\begin{table}[t]
  \begin{center}
    \begin{tabular}{rp{0.5cm}p{0.5cm}p{0.5cm}p{0.5cm}p{0.5cm}p{0.5cm}p{0.5cm}p{0.5cm}p{0.5cm}p{0.5cm}p{.7cm}p{.7cm}}
    \small
       & 0 & 1 & 2 & 3 & 4 & 5 & 6 & 7 & 8 & 9 & 10 & 11\\
    \normalsize
      $\St_{p,q}[i]\;\;$ & $\p$ & $\q$ & $\p$ & $\q$ & $\p$ & $\p$ & $\q$ & $\p$ & $\q$ & $\p$ & $\p/\q$ & $\q/\p$\\
      $\widetilde{G}(i,p,q)\;\;$ & $p$ & $q$ & $2p$ & $2q$ & $3p$ & $4p$ & $3q$ & $5p$ & $4q$ & $6p$ &  & \\
      $\widetilde{G}(i,p,q)\;\;$ & $5$ & $7$ & $10$ & $14$ & $15$ & $20$ & $21$ & $25$ & $28$ & $30$ &  & 
    \end{tabular}
  \end{center}\caption{The Sturmian words $\St_{p,q}$ for $p=5$ and $q=7$ and the corresponding values of $\widetilde{G}(i,p,q)$
  for $i<p+q-2$.}\label{tab:sturmian}
\end{table}

\begin{fact}[{\cite[Exercise 2.2.9]{Lothaire2002}}]\label{fct:pref}
$\St(\gamma'_0,\ldots,\gamma'_{m'})$ is a prefix of $\St(\gamma)$
if and only if $[0;\gamma'_0,\ldots,\gamma'_{m'}]$ is a semiconvergent of $\fr(\gamma)$.
\end{fact}

\begin{example}
The semiconvergents of $[0;1,2,2]=\frac{5}{7}=[0;1,2,1,1]$ are $[0;1,2,1]=\frac{3}{4}$, $[0;1,2]=\frac{2}{3}$, $[0;1,1]=\frac{1}{2}$, $[0;1]=1$,
 $[0;]=\frac{0}{1}$ (and $\frac{1}{0}$).
They correspond to the following Sturmian prefixes of $\St(1,2,2)=\p\q\p\q\p\p\q\p\q\p\p\q$:
$\St(1,2,1)=\p\q\p\q\p\p\p\q$, $\St(1,2)=\p\q\p\q\p$, $\St(1,1)=\p\q\p$, $\St(1)=\p\q$, and $\St()=\p$.
\end{example}

\begin{corollary}\label{cor:pred}
Consider a proper prefix $P$ of Sturmian word $\St(\gamma)$.
Moreover, let $\frac{a}{b}=\pred_{|P|}(\fr(\gamma))$ and $\frac{c}{d}=\succ_{|P|}(\fr(\gamma))$.
The longest even Sturmian prefix of $P$ has length $a+b$,
whereas the longest odd Sturmian prefix of $P$ has length $c+d$.
\end{corollary}
\begin{proof}
By \cref{fct:pref}, the longest even Sturmian prefix of $P$ is the longest Sturmian word $\St(\gamma')$
such that $\frac{a'}{b'}:=\fr(\gamma')$ is an even semiconvergent of $\fr(\gamma)$.
Its length $a'+b'\le |P|$ is largest possible, so by \cref{fct:farey} $\frac{a'}{b'}$ is
the  best left approximation of $\fr(\gamma)$ with $a'+b'\le |P|$.
This is precisely how $\frac{a}{b}=\pred_{|P|}(\fr(\gamma))$ is defined.

The proof for odd Sturmian prefixes is symmetric.
\end{proof}
The following theorem can be seen as a restatement of \cref{lem:ld} in terms of Sturmian words.
\begin{theorem}\label{thm:sturm}
  Let $\St_{p,q}$ be a standard Sturmian word corresponding to $\frac{p}{q}$ and let $0\le h< p+q-3$.
  If $\St_{p,q}[0..h+3]$ is a Sturmian word, then
  $L^d(h,p,q)=\widetilde{G}(l-2,p,q)+\widetilde{G}(r-2,p,q)$,
  where $l,r$ are the lengths of the longest proper Sturmian prefixes of $\St_{p,q}[0..h+3]$ of different parities,
  and $\widetilde{G}(-1,p,q)=0$.
Otherwise,  $L^d(h,p,q)=\widetilde{G}(h+2,p,q)$. 
\end{theorem}
\begin{proof}
To apply \cref{lem:ld}, we set $\frac{a}{b}=\pred_{h+3}(\frac{p}{q}$ and $\frac{c}{d}=\succ_{h+3}(\frac{p}{q})$.
Observe that the mediant $\frac{a+c}{b+d}$ is a better approximation of $\frac{p}{q}$ than $\frac{a}{b}$ or $\frac{c}{d}$,
and thus it is a semiconvergent of $\frac{p}{q}$.
Thus, we always have $a+b+c+d\ge h+4$ and, by \cref{fct:pref}, equality holds if and only if $\St_{p,q}$ has a Sturmian prefix of length $h+4$.
In other words, the case distinction here coincides with the one in \cref{lem:ld}.
If $a+b+c+d>h+4$, then we have $L^d(h,p,q)=\widetilde{G}(h+2,p,q)$.
Otherwise,  $L^d(h,p,q)=\widetilde{G}(a+b-2,p,q)+\widetilde{G}(c+d-2,p,q)$.
However,  due to \cref{cor:pred}, $\St_{p,q}[0..a+b-1]$ is an even Sturmian word corresponding to $(a,b)$,
$\St_{p,q}[0..c+d-1]$ is an odd Sturmian word corresponding to $(c,d)$,
and these are the longest Sturmian prefixes of $\St_{p,q}[0..h+2]$ of each parity.
\end{proof}

\begin{example}
  Consider a word $\St_{5,7}$ as in \cref{tab:sturmian}.
  The lengths of its proper even Sturmian prefixes are $2,7$, whereas the lengths of its proper odd Sturmian prefixes are $1,3,5$.
  Hence, $L^d(7,5,7)=\widetilde{G}(9,5,7)=30$, since $\St_{5,7}[0..10]$ is not a Sturmian word.
  Moreover, $L^d(8,5,7)=\widetilde{G}(5,5,7)+\widetilde{G}(3,5,7)=20+14=34$, since $\St_{5,7}[0..11]=\St_{5,7}$
  is a Sturmian word.
\end{example}

  \bibliographystyle{plainurl}
  \bibliography{periodicity}

\end{document}